\documentclass{easychair}

\usepackage{doc}

\usepackage{soul}
\usepackage{url}
\usepackage[utf8]{inputenc}
\usepackage[small]{caption}
\usepackage{graphicx}
\usepackage{amsmath}
\usepackage{amsthm}
\usepackage{booktabs}
\usepackage[switch]{lineno}

\newtheorem{example}{Example}
\newtheorem{theorem}{Theorem}

\usepackage{soul}
\usepackage{url}
\usepackage[utf8]{inputenc}
 \usepackage{graphicx}
\usepackage{amsmath}
\usepackage{amsthm}
\usepackage{amssymb}
\usepackage{amsfonts}
\usepackage{booktabs}
\usepackage{paralist}
\usepackage{stmaryrd}
\usepackage{thmtools}
\usepackage{turnstile}

\usepackage{float}
\restylefloat{table}

\usepackage{mathtools}
\usepackage{cuted}

\usepackage{tikz-cd}
\usepackage{float}
\usepackage{graphicx}
\usepackage{xspace}

\restylefloat{table}

\usepackage{tikz}

\usetikzlibrary{matrix,arrows.meta}

\usetikzlibrary{calc,arrows,shapes}
\usetikzlibrary{trees,fit,shapes,calc}

\tikzstyle{mybox} = [rectangle,draw,minimum height=1cm, minimum width=6em,fill=blue!20]
\tikzstyle{decision} = [diamond, draw,
text width=4.5em, text badly centered, node distance=4cm, inner sep=0pt]
\tikzstyle{block} = [rectangle, draw,
text width=5em, text centered, rounded corners, node distance=4cm, minimum height=4em]
\tikzstyle{line} = [draw, -latex']
\tikzstyle{cloud} = [draw, ellipse, node distance=4cm,
minimum height=2em]

\usepackage{mathtools}

\usepackage{cuted}
\usetikzlibrary{arrows.meta,shapes.misc,patterns,shapes.symbols,decorations.pathreplacing,decorations}

\DeclareMathSymbol{\mlq}{\mathord}{operators}{``}
\DeclareMathSymbol{\mrq}{\mathord}{operators}{`'}

\DeclareMathOperator{\dom}{dom}

\usepackage{tikz}
\usetikzlibrary{trees}

\tikzset{My Style/.style={blue, draw=black,fill=pink, minimum size=0.5cm}}
\colorlet{shadecolor}{gray!40}
\tikzset{My Style/.style={black, draw=white,fill=shadecolor, minimum size=0.1cm}} 

\newcommand{\ignore}[1]{} 

\newcommand{\defin}[1]{\left\{\begin{array}{l} #1 \end{array}\right\}}

\newcommand{\rul}{\leftarrowtail}

\newtheorem{corollary}{Corollary}
\newtheorem{proposition}{Proposition}

\theoremstyle{definition}
\newtheorem{definition}{Definition}
\newtheorem{notation}{Notation}

\ignore{

}

\newcommand{\vu}{\bar{u}}                

\newcommand{\bx}{\bar{x}}                

\newcommand{\bP}{\bar{P}}

\newcommand{\ba}{\bar{a}} 
\newcommand{\bb}{\bar{b}}

\newcommand{\eq}{=}

\newcommand{\PTIME}{P-time\xspace}

\newcommand{\op}{\textbf{op}}
\newcommand{\e}{\varepsilon}

\newcommand{\me}{\module(\e)}

\newcommand{\DomH}{{\Un^+}}
\newcommand{\DomTH}{{\Un^+}}

\newcommand{\Pa}{\cP_{\term}}

\newcommand{\GuessP}{\textit{GuessP}}
\newcommand{\GuessNewP}{\textit{GuessNewP}}
\newcommand{\GuessNewPair}{\textit{GuessNewPair}}
\newcommand{\Copy}{\textit{Copy}}
\newcommand{\Reach}{\textit{Reach}}
\newcommand{\Pick}{\textit{Pick}}
\newcommand{\PickPQ}{\textit{PickPQ}}
\newcommand{\Eq}{\textit{Eq}}
\newcommand{\Elim}{\textit{Elim}}
\newcommand{\Root}{\textit{Root}}

\newcommand{\BG}{\textbf{\tiny{ BG}}}

\newcommand{\partialmap}{\rightharpoonup}

\newcommand{\Lo}{\mathbb{L}}

\newcommand{\Dom}{{\rm Dom}}

\newcommand{\id}{{\rm id}}

\newcommand{\Un}{{\bf U}}

\newcommand{\mT}{{\top}}  

\newcommand{\rneg}{\mathord{\curvearrowright}}

\ignore{
\newcommand\turnstile[3][2]{%
	\mathrel{\calcturnstile{#1}%
		\setbox0=\hbox{$\tsvstyle\vdash$}%
		\setbox2=\hbox{$\vcenter{\copy0}$}%
		\hbox{\vrule\vphantom{$\tsvstyle\vdash$}}%
		\raise\dimexpr\ht0-\ht2\relax\hbox{$\vcenter{\offinterlineskip
				\ialign{\hfil\kern1pt$\tsstyle##\vphantom{by}$\kern1pt\hfil\cr
					\relax\if\relax\detokenize{#2}\relax\hphantom{\vdash}\kern-2pt\else#2\fi\cr
					\noalign{\kern1pt\hrule\kern1pt}%
					\relax\if\relax\detokenize{#2}\relax\hphantom{\vdash}\kern-2pt\else#3\fi\cr}}$}%
}}
\newcommand\Turnstile[3][2]{%
	\mathrel{\calcturnstile{#1}%
		\vcenter{\hbox{\vrule\vphantom{$\tsvstyle\models$}}}%
		\vcenter{\offinterlineskip
			\ialign{\hfil\kern1pt$\tsstyle##\vphantom{by}$\hfil\cr
				\relax\if\relax\detokenize{#2}\relax\hphantom{=}\kern-2pt\else#2\fi\cr
				\noalign{\kern1pt\hrule\kern\fontdimen22\tsfont2\hrule\kern1pt}%
				\relax\if\relax\detokenize{#2}\relax\hphantom{-}\kern-2pt\else#3\fi\cr}}
}}
\newcommand\calcturnstile[1]{%
	\ifcase#1\let\tsstyle\textstyle\let\tsvstyle\textstyle\let\tsfont\textfont\or
	\let\tsstyle\textstyle\let\tsvstyle\textstyle\let\tsfont\textfont\or
	\let\tsstyle\scriptstyle\let\tsvstyle\textstyle\let\tsfont\scriptfont\or
	\let\tsstyle\scriptscriptstyle\let\tsvstyle\scriptstyle\let\tsfont\scriptscriptfont\fi}
} 

\ignore{
 \newcommand*\Subsststile[2]{%
\ \	\,\scalebox{0.8}[0.5]{$\sststile[ss]{\textstyle#1}{\textstyle#2}$}\, \ \
}
\newcommand{xves}{\Subsststile{}{\ \ \ \ch \ \ \ }}
}

\newcommand{\sststiles}[3][s]{\turnstile[#1]{s}{s}{#2}{#3}{s}}

 \newcommand*\mysststile[2]{%
	\ \	\,\scalebox{0.8}[0.5]{$\sststile[ss]{\textstyle#1}{\textstyle#2}$}\, \ \
}
\newcommand*\mydoublesststile[2]{%
	\ \	\,\scalebox{0.8}[0.5]{$\sststiles[ss]{\textstyle#1}{\textstyle#2}$}\, \ \
}

\newcommand{\provesxx}{\mysststile{}{\ \ \ x\ \ \ }}

\newcommand{\provesxy}{\mysststile{}{\ \ \ xy\ \ \ }}

\newcommand{\doubleprovesxy}{\mydoublesststile{}{\ \ \ \ xy \ \ \ }}

\newcommand{\doubleprovesxyzw}{\mydoublesststile{}{\ \  \ xyzw\ \ }}

\newcommand{\provesxyzw}{\mysststile{}{\  xyzw \ }}

\newcommand{\iter}{\uparrow}
\newcommand{\comp}{\ensuremath \mathbin{;}}

\newcommand{\dn}{\!\downarrow}
\newcommand{\notdn}{\!\not\downarrow}

\newcommand{\sstring}{s}

\newcommand{\nat}{\mathbb{N}}  

\newcommand{\strA}{\fA}
\newcommand{\strB}{\fB}         
\newcommand{\strC}{\fC}  
  
\newcommand{\CH}{{\sf CH} }

\newcommand{\ChFunctions}{\CH(\Un,\Modules)}

\newcommand{\cP}{\mathcal{P}}

\newcommand{\cF}{\mathcal{F}}

\newcommand{\Mod}{\mathcal{M}}

\newcommand{\cH}{\mathcal{H}}

\newcommand{\1}{{\bf 1}}

\newcommand{\fA}{\mathfrak{A}}
\newcommand{\fB}{\mathfrak{B}}

\newcommand{\fC}{\mathfrak{C}}

\newcommand{\cK}{\mathcal{K}}

\newcommand{\Tree}{{\cal T}}
\newcommand{\Tr}{\textbf{Tr}}

\newcommand{\llast}{\mathit{last}}
\newcommand{\ffirst}{\mathit{first}}

\newcommand{\Domain}{\mathcal{D}}
\newcommand{\Last}{\mT}

\newcommand{\emp}{\mathbf{e}}

\newcommand{\withconverse}[1]{}

\newcommand{\term}{t}
\newcommand{\termg}{g}

\newcommand{\ch}{h}
\newcommand{\che}{\bar{h}}

\newcommand{\semaNoh}[1]{\textbf{Tr}\llbracket{#1}\rrbracket}
\newcommand{\sem}[1]{\textbf{Tr}\llbracket{#1}\rrbracket}

\newcommand{\module}{m}
\newcommand{\Modules}{\Mod}

\newcommand{\Terms}{\it Terms}

\newcommand{\ioeq}{\fallingdotseq}
\newcommand{\strongeq}{\doteq}
\newcommand{\traceeq}{\doteq}
\newcommand{\Proj}{{\sf Pr}} 

\makeatletter
\tikzset{
	dot diameter/.store in=\dot@diameter,
	dot diameter=3pt,
	dot spacing/.store in=\dot@spacing,
	dot spacing=10pt,
	dots/.style={
		line width=#1,
		line cap=round,
		dash pattern=on 0pt off \dot@spacing
	},
	mydots/.style 2 args={
		dot spacing=#2,
		preaction={draw=black,
			dots=#1},
		draw = white,
		dots=.8*#1,
		postaction={draw=black,
			dots=.333*#1},
	}
}
\makeatother

\title{Promise Algebra:\\ An Algebraic Model of Non-Deterministic
	 Computations}

\author{
Eugenia Ternovska
}

\institute{
	Simon Fraser University\\
	\email{ter@sfu.ca}
}

\authorrunning{Eugenia Ternovska}

\titlerunning{Promise Algebra}

\begin{document}

\maketitle

\begin{abstract}
Our goal is to define an algebraic language for reasoning about non-deterministic computations. Towards this goal, 
we introduce an algebra of string-to-string transductions. 
Specifically, it is an algebra of partial functions on words over the alphabet of relational $\tau$-structures over the same domain.
The algebra has a two-level syntax, and thus, two parameters to control its expressive power. The top level defines algebraic expressions, and the bottom level specifies atomic transitions. 
History-dependent  Choice functions resolve atomic non-determinism, and make general relations functional. Equivalence classes  of such functions  serve as certificates for computational problems specified by algebraic terms.
The algebra has an equivalent syntax in the form of a Dynamic Logic, where terms describing computational processes or programs appear inside the modalities. 

We define a  simple secondary logic for representing atomic transitions, which is  a modification of conjunctive queries. With this logic, the algebra can represent both reachability and counting examples, which is not possible in Datalog.  
We analyze  the data complexity of this logic, measured in the size of the input structure, and show that a restricted fragment of the logic captures the complexity class NP. 

The logic can be viewed as a database query language, where 
atomic propagations are separated from control.

\end{abstract}

\vspace{1ex}

\noindent\textbf{Organization}  
We start by definining the syntax and semantics of the algebra.
In {Section  \ref{sec:Algebra}}, we present the syntax of the algebra, and in Section \ref{sec:Semantics}, its semantics. In {Section  \ref{sec:Dynamic-Logic}}, we  reformulate  the algebra as a modal  Dynamic  Logic. The logic allows for complex nested  tests, and has an iterator construct.  The main programming constructs are definable in the logic.
The formulae have a free function variable  over Choice functions, which gives us an implicit existential quantifier  over Choice functions. 
Due to the presence of a form of negation, an implicit universal quantification  over such functions is also present. Moreover, in complex nested tests, such implicit (existential, universal) quantifiers can alternate.

In Section \ref{sec:main-task}, we define our main computational task, formulated in the Dynamic Logic.
We then define the notion of a computational problem specified by an algebraic term.  A certificate for such a problem is an equivalence class of Choice functions, also called a witness or a \emph{promise}. We show that a Boolean algebra of promises is embedded into the Dynamic Logic, and its underlying set has a forest structure, where the forest has a tree for each term.  In Section \ref{sec:Strong-Equalities}, we show that the truth of certain (conditional) equalities between algebraic terms can be used to indicate the existence, or non-existence, of a ``yes'' certificate for a problem specified by an algebraic term.
In Section \ref{sec:Atomic-Modules}, we explain how a secondary logic can be defined. In particular, we formulate 
the Law of Inertia, and define a specific logic for specifying atomic transitions. The logic is based on a modification of unary conjunctive queries. 
We give examples in  {Section \ref{sec:Examples}}, and study the complexity of query evaluation in 
{Section \ref{sec:Complexity}}.  
We conclude, in {Section \ref{sec:Conclusion}}, with a summary and future research directions. 
Related work is mentioned throughout the paper.

\vspace{1ex}

We assume familiarity with the basic notions of first-order (FO) and second-order (SO) logic  (see, e.g., \cite{Enderton}), and use `$:=$' to mean ``denotes'' or ``is by definition''.

\newpage
\section{Algebra: Syntax}
\label{sec:Algebra}

In this section, we define the syntax of our algebra, and introduce some initial intuitions regarding the meaning of its operations. The purpose of the algebra is to talk about Choice functions of a certain kind. The Choice functions (or, rather, their equivalence classes) will later  be considered as certificates for computational problems specified by algebraic terms, or \emph{promises}. If such a promise is given for a term (intuitively, a non-deterministic program), the computation is guaranteed to proceed successfully. With this intuition in mind, we call this algebra a \emph{Promise Algebra}. A promise is like an elephant in a room -- it is there, but is never mentioned explicitly.\footnote{We use a free function \emph{variable} $\e$ to denote the presence of this elephant, but no concrete certificate is present in the language.}

Formally, the algebra is an algebra of (functional) binary relations on strings of relational structures.  It has a two-level syntax. The top level (defined in this section) specifies algebraic expressions. The algebraic operations are \emph{dynamic counterparts }of classical logic connectives --  negation, conjunction, disjunction -- and iteration (the Kleene star, or reflexive transitive closure). Unlike classical connectives, these operations are \emph{function-preserving} in the sense that if atomic elements are functional binary relations, then so are all algebraic expressions built up from them. The bottom level of the formalism specifies atomic expressions (intuitively, actions) in a separate logic (explained later). 

We now remind the reader the notion of a (relational) structure.
Let  $\tau$ be a relational vocabulary, which is finite (but is of an unlimited size).  
Let $\tau \ := \ \{S_1, \dots , S_n\}$,  each $S_i$ has an associated arity $r_i$,  and $A$ be a non-empty set.
A \emph{$\tau$-structure} $\strA$ over  domain $\dom(\strA)\  : = \ A$  is 
$
\strA \ := \ (A;\  S_1^{\strA}, \dots , S_n^{\strA} ),
$
where $S_i^{\strA}$ is an $r_i$-ary relation  called the \emph{interpretation} of  $S_i$.
In this paper, all structures are finite.
If $\strA$ is a $\tau$-structure,  $\strA|_{\sigma}$ is
its restriction to a  sub-vocabulary $\sigma$. 
We now fix a relational vocabulary  $\tau$, and assume it is partitioned into ``inputs'' (or EDB relations, in the database terminology), and unary ``register'' symbols:\footnote{EDB is common term in Database theory. It stands for Extensional Database, which is a relational structure.}   \begin{equation}\label{eq:tau}
	\tau \ := \ \tau_{\rm EDB}   \uplus 
	\tau_{\rm reg}.
\end{equation}
The details and the formal requirements on these two vocabularies will become clear when a particular logic of the bottom  level is explained later in the paper. We only mention now that, intuitively, the interpretations of EDB (or input) relations never change, while the interpretations of 
the registers are updated by applications of atomic modules.

Let a set $\Modules$ of atomic module symbols, denoting non-deterministic atomic actions,  be fixed. Intuitively, the actions are  atomic updates of relational structures. The atomic updates are combined into complex updates (formally, algebraic expressions, or terms) using algebraic operations.
The set $\Terms$  of the well-formed terms of the algebra are defined as:
\begin{equation}	\label{eq:algebra}  
	\begin{array}{c}	
	\hspace{-3ex}	\hspace{-1ex}	\term  :: =   
		\me \mid	\id   \mid    \rneg\term   \mid  
		\term\comp \term  \mid 
		\term  \sqcup \term  \mid    \term ^\iter  
		\mid\! P= Q\!	\mid \!  \BG (P_{now}\neq\! Q),
	\end{array}
\end{equation}
where $\module \in \Modules$ and $\e$ is a free function variable ranging over Choice functions. The syntax allows  only \emph{one} such (free) variable per an algebraic expression.\footnote{Syntactically, the variable $\e$, that occurs free in the algebraic expressions, is not really necessary. We use it to emphasize the \emph{existence} of a certificate -- a concrete Choice function. It will be convenient when we formalize a computational problem specified by an algebraic term.  }
 of all relational $\tau$-structures over the same finite domain.   
We require that $P,Q \in   \tau_{\rm reg}$. The subscript ``$now$'' in  $P_{now}$ is not a part of the syntax. It is added for the reader to memorize that the content of the ``register'' $P$   in the current  position must be different from the values  in $Q$ ever before  (denoted `$\BG$', for Back Globally).

	In the table, we list the variables first, and then the  partial functions of the algebra that are split into three groups -- nullary (constant), unary and binary partial mappings. Unary functions take one other function as an argument, binary-- two. The column in the middle represent the range of the variable or the type of the corresponding partial function.  Our notations are:  $\Un$ is the set of all $\tau$-structures over the same domain,
	$\ChFunctions$ is a set of Choice functions, to be defined later in the paper,
and	$\cF$ denotes partial functions on $\DomH$ (non-empty strings over the alphabet $\Un$).
 
\vspace{1ex}

\begin{center}
\begin{small}
\hspace{-3ex}\begin{tabular}{c|c|c}
		$x$& $\DomH$ & string variable\\
			$\e$&  $\ChFunctions$  & Choice function variable\\
				$f,g$& $\cF$  & partial function variables on  $\DomH$ \\
		\hline 
	\hline 
		  $\id(x)$&&Identity on $\DomH$\\
	$\module(\ch/\e)(x)$&& Module, for all $\module\in \Modules$, $\ch\in\cH$ \\
	$\BG(P\neq Q) (x)$&$ \cF$&Back Globally Non-Equal  \\
	$(P= Q)(x)$&& Equal  \\
	\hline 
	$\rneg f(x)$&$\cF\to\cF$&Anti-Domain (Unary Negation)\\
	$f^{\iter}(x)$& &Maximum Iterate\\
	\hline 
	$(f\sqcup g)(x)$&$\cF^2\to\cF$&Preferential Union\\
	$(f\comp g)(x)$& &Sequential Composition \\
	\hline 
	\hline 
\end{tabular} 
\end{small}

\end{center}

Intuitively, partial mappings $\cF$ on $\DomTH$ represent programs. Applying such a mapping on a string of length one corresponds to applying a program on an input structure, e.g., a graph in 3-Colourability problem. The  set $\cF$ of partial functions on $\DomTH$ contains all functional constant (i.e., nullary) operations (the rows marked with $\cF$), and  is closed under unary (marked $\cF\to\cF$) and binary (marked $\cF^2\to\cF$) operations.

\vspace{1ex}

	\noindent \underline{Nullary (constant mappings on strings in $\DomH$):}  \ \ \ \ 
	$\id \in \cF$, 
	$ \module(\ch/\e)\in \cF$ 
	for all  $\module$ in  $\Modules$ and $\ch \in \ChFunctions$, and 	$\{(P=Q),\BG(P\neq Q) \}\subset\cF$,
	for all $P,Q\in \tau_{\rm reg}$.
	Here, $\id$ is the Identity function (Diagonal relation) on strings in $\Un^+$. 
Each $\module(\ch/\e)$, $\module\in \Modules$,  is an atomic module (action), a binary relation, that, intuitively, updates some of the registers in $\tau_{\rm reg}$, non-deterministically, i.e., multiple outcomes of an action $ \module(\e)$  are possible. With an instantiation of a concrete Choice function $\ch$ in $\module(\ch/\e)$, the relation $ \module(\e)$ becomes a partial  function function in $\cF$. 
Back Globally Non-Equal,  denoted $\BG(P\neq Q)(x)$, checks if 
the value stored in register $Q$ at any point earlier in string $x$ is \emph{different}  from the value in $P$ now.
Equality check $(P=Q)$ compares current interpretations of $P$ and $Q$.  Thus, access to domain elements is allowed only via  atomic updates (modules in $\Modules$ specified in a secondary logic) and (in)equality checks only.

\vspace{1ex}

\noindent \underline{Unary $\op_1 : \cF \to \cF$} (partial mappings on strings in $\DomH$ that depend on one partial function):
$ \op_1 \in \{ \rneg\,, ^{\iter}\} $.  Each operation $ \op_1$ in this set takes a function in $\cF$ and modifies it according to the semantics of  $ \op_1$. The modified operation is applied on strings in $\Un^+$.
Anti-Domain $\rneg f(x)$ checks if there is no outgoing $f$-transition from the state represented by the string $x$; 

Maximum Iterate  $f^{\iter}(x)$  is a ``determinization'' of  the Kleene star $f^*(x)$  (reflexive transitive closure). It 
outputs only the longest, in terms of the number of $f$-steps, transition out of all possible transitions from the same state produced 
by the Kleene star.\footnote {For readers familiar with the $\mu$ operator, we mention that $f^{\iter}\ :=\ \mu Z. ( \rneg f \cup f \comp Z)
$, although these constructs are not in our language.}

\vspace{1ex}

\noindent \underline{Binary $\op_2 : \cF^2 \to \cF$} (partial mappings on strings in $\Un^+$ that depend on two functions):
$
\op_2 \in \{ \comp,   \sqcup \}
$
These operations take two functions as arguments, and combine them to obtain a (partial) mapping on $\Un^+$.
Sequential Composition  $(f\comp g)(x)$ is the standard function composition  $g(f(x))$.  Preferential Union $(f\sqcup g)(x)$ applies $f$ and returns its result; but, if $f$ is not defined, it applies $g$.

We will routinely omit string variable $x$ in algebraic terms and work entirely with partial  mappings in $\cF$ on $\DomTH$. This is a typical notational convention in  algebraic setting such as groups, algebras of functions and binary relations, etc.

 \section{Semantics} 
 \label{sec:Semantics}

 \subsection{Atomic State-to-State Transitions}
  The semantics of the algebra is parameterized by an underlying transition system that specifies  structure-to-structure transitions associated with each atomic module symbol. 
 Formally,  we have a transition relation  $\sem{\cdot}$ that maps each atomic module symbol in $\Modules$ to a binary relation  on $\Un$ (the set of all relational $\tau$-structures over the same domain):
 $$
 \sem{\cdot}: \Modules  \to \Un\times\Un.
 $$
  In general, the relation $\sem{\cdot}$ is specified in any  (secondary) logic that constitutes the bottom level of the algebra.
 In Section \ref{sec:Atomic-Modules}, we give a specific example of such a logic. But, for now, it is sufficient to think of $\sem{\cdot}$  as an arbitrary binary relation on structures in $\Un$.
    The relation $\semaNoh{\module(\e)} \subseteq \Un \times \Un$, for $\module\in \Modules$,  is not necessarily functional. In this sense, atomic transitions are non-deterministic.
     But,  specific Choice functions (to be defined shortly), that  are given semantically only,   resolve atomic non-determinism, while taking the history into account.\footnote{Notice that concrete Choice functions are not a part of the syntax. The syntax, summarized in (\ref{eq:algebra}) and detailed in the table, has only a Choice function variable $\e$.  The terms impose constraints on possible Choice functions. }  
   To explain it formally, we first give a definition of a tree.

  \vspace{1ex} 
  
  A \emph{tree} over alphabet $\Un$ is a (finite or infinite) nonempty set $\Tree\subseteq \Un^*$ such that for all $x\cdot c\in \Tree$, with $x\in \Un^*$ and $c\in \Un$, we have $x\in \Tree$.
  The elements of $\Tree$ are called \emph{nodes}, and the empty word $\mathbf{e}$ is the \emph{root} of $\Tree$. For every $x\in \Tree$, the nodes $x\cdot c\in \Tree$ where $c\in \Un$ are the \emph{children} of $x$.  A node with no children is a \emph{leaf}.

  \vspace{-1.2ex}
  
  \newcounter{inlineequation}
  \setcounter{inlineequation}{0}
  \renewcommand{\theinlineequation}{(\Roman{inlineequation})}
  
  \newcommand{\inlineeq}[1]{\refstepcounter{inlineequation}\theinlineequation\ \(#1\)}

  \let\sse=\subseteq
  \let\vf=\varphi
  \def\vc#1#2{#1 _1\zd #1 _{#2}}
  \def\zd{,\ldots,}

  \colorlet{shadecolor}{gray!30}

  \tikzset{My Style/.style={black, draw=shadecolor,fill=shadecolor, minimum size=0.1cm}}

  \begin{figure}[H]
  	
  	\hspace{15ex} 
  	\begin{tikzpicture}[
  		grow= right,
  		level distance=1.2cm,
  		level 1/.style={sibling distance=1.0cm},
  		level 2/.style={sibling distance= 0.6cm},
  		level 3/.style={sibling distance= 0.6cm},
  		level 4/.style={sibling distance= 0.58cm},
  		level 5/.style={sibling distance= 0.38cm}]
  		
  		\tikzset{level 1 onwards/.style={level distance=0.2cm}}
  		\tikzset{level 2 onwards/.style={level distance=0.2cm}}
  		\tikzset{level 3 onwards/.style={level distance=0.4cm}}

  		\node (Root)(z) {$\emp$}
  		child {
  			node (a)[left=0.2cm] {$\strB$} 
  			child { node (b)[left=3mm] {} edge from parent[dashed] 
  				node[right,draw=none]{} 
  			}
  			child { node (c)[left=3mm] {} edge from parent[dashed] 
  				node[right,draw=none]{} }
  			edge from parent node[left,draw=none] {}
  		}
  		child {
  			node (d) [left=0.2cm] 
  			{$\strA$} 
  			child { node (e) {$\strA\cdot\strB$}
  				child { node(f) [right=0.2mm]
  					{$\strA\cdot\strB\cdot\strB$}
  					child { node(g)[right=1mm] {$\strA\cdot\strB\cdot\strB\cdot\strB$} 
  					edge from parent node[right,draw=none]{}}
  				child { node(h)[right=1mm][My Style] {$\strA\cdot\strB\cdot\strB\cdot\strA$}
  		  					edge from parent node[left,draw=none]{} }
  				edge from parent
  				node[left,draw=none]{}}
  			child { node(x)[right=0.2mm] {$\strA\cdot \strB\cdot\strA$} 
  				child { node(m) {} edge from parent[dashed]
  					node[right,draw=none]{}}
  				child { node(q) {} edge from parent[dashed]
  					node[left,draw=none]{}}
  				edge from parent node[right,draw=none]{}}
  		}
  		child { node (o){$\strA\cdot \strA$}
  			child { node(p)[left=1mm] {} edge from parent[dashed]
  				node[right,draw=none]{}}
  			child { node (q)[left=1mm] {} edge from parent[dashed] 
  				node[right,draw=none]{} }
  		}
  	};
  	
  	\tikzset{every loop/.style={min distance=10mm,in=160,out=120,looseness=10}}
  	
  	\tikzset{every loop/.style={min distance=10mm,in=160,out=120,looseness=10}}
  	
  	\draw[ultra thick,gray] (z)--(d);
  	\draw[ultra thick,gray] (d)--(e);
  	\draw[ultra thick,gray] (e)--(f);
  	\draw[ultra thick,gray] (f)--(h);

  \end{tikzpicture}
  \caption{A tree over $\Un$.  Let $\module_{17}\in \Modules$ and $(\strA,\strB) \in \sem{\module_{17}(\e)}$, $(\strB,\strA) \in \sem{\module_{17}(\e)}$,  $(\strB,\strB) \in \sem{\module_{17}(\e)}$. A concrete Choice function $
  	\ch: \Modules  \to  (\Un^+ \partialmap \Un^+)$ maps
  	$\module_{17}(\e)$   to a set of mapping containing, for example, $\strA\cdot\strB\cdot\strB \mapsto \strA\cdot\strB\cdot\strB\cdot\strA$, and other such mappings shown by the thicker edges in the tree. Term $\term \ := \ \module_{17}\comp \module_{17}\comp\module_{17}(\ch/\e)$, with this specific $\ch$, makes transitions along the branch in this tree shown by the thicker edges. There is a one-to-one correspondence between that branch and the string shown by the shaded area, which is a \emph{trace} of $\term$ from the input structure $\strA$ that corresponds to $\ch$.  } 
\label{fig:tree}
\end{figure}

\vspace{-2.2ex}

\vspace{1ex}

\noindent\textbf{Notations for Strings and Partial Functions} 
We use $\emp$ to denote the empty string;
$\sstring(i)$, for $i\geq 1$,  to denote the $i$'s letter in string $\sstring$; and $\sstring_i$ to denote the prefix of string $s$ ending in position $i$. In particular, $\sstring_0=\emp$.
We use the following notation for partial functions:
$A \partialmap B \ :=  \  \bigcup_{C\subseteq A}   (C\to B)$. 

\vspace{3ex}

\subsection{Choice Functions} 
Each prefix $\sstring_i$ of a string $\sstring$ represents a history, i.e., a sequence of states -- elements of $\Un$. For each history  $\sstring_i$, Choice function $\ch$ selects one possible outcome of an atomic module $\module(\e)$, out of those possible, i.e., those in its interpretation $ \sem{\module(\e)}$, see Figure 1. It decides $\module(\e)$'s outcome at $\sstring_i$, if such an outcome is possible according to the transition system $ \sem{\cdot}$.  The outcome is history-dependent, that is, for the same module $\module(\e)$, we may have different outcomes at different time points $\sstring_i$ and $\sstring_j$, $i\neq j$. We now present this intuition formally.

\begin{definition}\label{def:Choice-function-atomic-modules}
	Let $\Modules$ be a finite set of module symbols, and $\Un$  be an alphabet.
A \emph{Choice function}  is a function 
$
\ch: \Modules  \to  (\Un^+ \partialmap \Un^+)$
such that
\begin{itemize}
		\item 	$	\ch(\module(\e)) = \{ (\sstring_i \mapsto \sstring_{i+1})  \mid   
	i\geq 1   \text{  and  exists }   (\sstring(i),\sstring(i+1))  \in  \sem{\module(\e)} 
	\text{ or }    i=0   \text{ and there exists }  \sstring(i+2) \text{ such that  } (\sstring(i+1),\sstring(i+2))  \in  \sem{\module(\e)}  \}$. 
\end{itemize}
\end{definition}
\noindent Note that $\ch$ outputs partial functions $\Un^+ \partialmap \Un^+$, so it resolves atomic non-determinism. Observe that, for the same module, different choices can be made at different time points because  we merely require the existence of a transition 
in $\sem{\module(\e)}$.

The second case,  $i=0$, is included so that we can have one tree (cf. Figure \ref{fig:tree}) per each term, with branches that correspond to different Choice functions, including the choice of allowable input structures. 
 More specifically, when $i=0$, we have  $\sstring_i=\emp$.  According to the definition above, if $(\strA,\strB) \in \sem{\module(\e)}$, then 	$(\emp\mapsto \strA) \in \ch(\module(\e))$, where $\strA$ is a possible input structure, see Figure 1.
Note  that 
$\sstring(i)=\sstring(i+1)$ is allowed -- a ``no-change'' transition
extends the string $\sstring_i$  by repeating the same letter in $\Un$. 

\vspace{1ex}

The set of all Choice functions is denoted $\ChFunctions$.

\vspace{1ex}

To summarize, a concrete Choice functions, e.g., $\ch$, uses the transition relation $\Tr$ to assign semantic meaning to an atomic expression in $\Modules$ as a set of \emph{string-to-string transductions}, that is, as a (partial) function on $\Un^+$.  Each $\module(\e)$ corresponds to some non-deterministic transitions in a transition system $\Tr$. A concrete Choice function, as a semantic  instantiation of $\e$, resolves this non-determinism, by taking the history into account. It extends the history (a string in $\Un^+$) by one letter in the alphabet $\Un$ that corresponds to one of the possible outcomes of $\module(\e)$  according to $\Tr$.
Intuitively, at each time step, $\ch$ resolves atomic non-determinism, while taking the history, that starts with an input structure, into account.  Thus, at different time points, there could be different outcomes.

\colorlet{shadecolor}{gray!50}
\tikzset{My Style/.style={black, draw=shadecolor,fill=shadecolor, minimum size=0.1cm}}

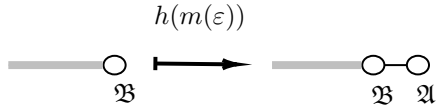
\begin{figure}[H]
	\begin{center}
		\tikzset{every picture/.style={line width=0.75pt}} 

\begin{tikzpicture}[x=0.75pt,y=0.75pt,yscale=-.5,xscale=.5]
	
	\draw [line width=1.5]    (264.33,72.8) -- (343.33,73.64) ;
	\draw [shift={(346.33,73.67)}, rotate = 180.6] [color={rgb, 255:red, 0; green, 0; blue, 0 }  ][line width=1.5]    (14.21,-4.28) .. controls (9.04,-1.82) and (4.3,-0.39) .. (0,0) .. controls (4.3,0.39) and (9.04,1.82) .. (14.21,4.28)   ;
	\draw [shift={(264.33,72.8)}, rotate = 180.6] [color={rgb, 255:red, 0; green, 0; blue, 0 }  ][line width=1.5]    (0,6.71) -- (0,-6.71)   ;

	\draw   (212,74.4) .. controls (212,69.21) and (217,65) .. (223.17,65) .. controls (229.33,65) and (234.33,69.21) .. (234.33,74.4) .. controls (234.33,79.59) and (229.33,83.8) .. (223.17,83.8) .. controls (217,83.8) and (212,79.59) .. (212,74.4) -- cycle ;

	\draw   (518,74.4) .. controls (518,69.21) and (523,65) .. (529.17,65) .. controls (535.33,65) and (540.33,69.21) .. (540.33,74.4) .. controls (540.33,79.59) and (535.33,83.8) .. (529.17,83.8) .. controls (523,83.8) and (518,79.59) .. (518,74.4) -- cycle ;
	\draw   (473,74.4) .. controls (473,69.21) and (478,65) .. (484.17,65) .. controls (490.33,65) and (495.33,69.21) .. (495.33,74.4) .. controls (495.33,79.59) and (490.33,83.8) .. (484.17,83.8) .. controls (478,83.8) and (473,79.59) .. (473,74.4) -- cycle ;
	\draw    (495.33,74.4) -- (518,74.4) ;
	\draw [line width=3]    (116,74.4) -- (211,74.4) [My Style];
	\draw [line width=3]    (382,74.4) -- (472,74.4) [My Style];

	\draw (260,10) node [anchor=north west][inner sep=0.75pt]   [align=left] {$\displaystyle \ch(\module(\e))$};
	\draw (219,90) node [anchor=north west][inner sep=0.75pt]   [align=left] {$\strB$};
	\draw (525,91.67) node [anchor=north west][inner sep=0.75pt]   [align=left] {$\strA$};
	\draw (479,92) node [anchor=north west][inner sep=0.75pt]   [align=left] {$\strB$};

\end{tikzpicture}
	\end{center}
	\caption{Choice function $\ch$ makes a specific selection $(\sstring\strB \mapsto \sstring\strB\strA)$,  provided that $(\strB,\strA) \in \sem{\module(\e)}$, out of all available possibilities $(\sstring_i \mapsto \sstring_{i+1})$  according to $\sem{\module(\e)}$. In set-theoretic terms, it has the standard meaning of picking an element out of a set of elements. The elements, in our case, are mappings on strings. }
	\label{fig:string-to-string-transition}
\end{figure}
Each algebraic term imposes constraints on allowable Choice functions, just as any first-order formula imposes constraints on possible instantiations of its free variables. The formal definition of the extension of Choice functions to all terms is given next. Recall that only one free variable $\e$ is allowed per term.

\subsection{Extension of Choice Functions to All Terms}
\label{sec:extension-of-Choice-functions}

So far, Choice functions were defined on the domain $\Modules$ of atomic module symbols. They are now   extended to operate on all terms, to define their semantics. 
Formally, 
\begin{equation}\label{eq:terms-partial-maps}
	\che: \Terms  \to  (\Un^+ \partialmap \Un^+).
\end{equation}
The extension   $\che$ of $\ch$ provides semantics to all  algebraic terms as string-to-string transductions.

\vspace{1ex}

\begin{enumerate}
	
	\item $\che(\module(\e)) =\ch(\module(\e))$ for $\module \in \Modules$.

	\item 
	$\che(\id)  = \{( \sstring_i\mapsto \sstring_i) \mid i \in \nat \} $.

	\item $\che(\rneg \term) = \{(\sstring_i \mapsto \sstring_i) \mid $  there is no string $s'$, Choice function  $\ch' $  and $k\geq i$ such that  $(\sstring_i \mapsto \sstring_k') \in \che'(\term)$, where $\sstring'_i=\sstring_i \}$.

	\item  $\che(\term\comp \term') = \{(\sstring_i \mapsto \sstring_j) \mid $   if there exists $l$, with $i\leq l \leq j$, such that  $(\sstring_i\mapsto \sstring_l) \in \che(\term)$ and  $(\sstring_l\mapsto \sstring_j) \in \che(\term')\}$.
	
	\item $\che(\term  \sqcup \term') = \{(\sstring_i \mapsto \sstring_j) \mid $  
	there exists $\che'$ such that $(\sstring_i\mapsto \sstring_j) \in \che'(\term)$ 
	and $\che(\term  \sqcup \term') = \che'(\term) $  or\\
	there is no $\che'$ such that $(\sstring_i\mapsto \sstring_k) \in \che'(\term)$ for any  $k$ 
	and 
	there exists $\che''$ such that $\che(\term  \sqcup \term') = \che''(\term')\}$.

	\item 	$\che(\term^\iter) = \{(\sstring_i \mapsto \sstring_j) \mid $  $i=j$ and $( \sstring_i \mapsto \sstring_j)\in \che(\rneg \term)$ \ \ \ or \ $i<j$ \ and 
	
	(1) 
	there exists  $l $ with $i\leq l \leq j$  such that $ (\sstring_i \mapsto  \sstring_l) \in \che(\term^\iter)$  and $(\sstring_l\mapsto \sstring_j) \in \che(\term)$, and  
	
	(2) there is no $ k$ with $  l < k \leq j $  and $\sstring(l) = \sstring(k)\}$. Thus,   the interpretation of $\term^\iter$  should not induce loops in the transition graph given by the binary relation $\sem{\cdot}$.\footnote{This is similar to how a transitive closure of a graph is defined.  The graph, in our case, is given by the transition system $\sem{\cdot}$. In other applications, this condition may be omitted.}
	
	\item  $\che(P = Q) = \{(\sstring_i \mapsto \sstring_i) \mid $  
	$Q^{s(i)} = P^{s(i)}\}$.
	
	\item $\che(\BG(P\neq Q)) = \{(\sstring_i \mapsto \sstring_i) \mid $    there is no   $1\leq l< i$ such that 
	$Q^{\sstring(l)} = P^{\sstring(i)}\}$.

	Observe that  Back Globally (\!$\BG(P\neq Q)$)  may ``look back'' through the entire string, not just a computation of a specific sub-term.
	\footnote{Note that the intended use of the algebra is to evaluate algebraic expressions with respect to an  input structure,  that is a one-letter string. In this use, the $\BG$ construct never looks into the ``pre-history'', 
			 since, in that use,  all strings of length greater than one are traces of some processes.}
	\footnote{An alternative definition of \!$\BG(P\neq Q)$ is possible, where the scope of  the condition is specified  explicitly (by e.g., adding parentheses). But, it is not needed since the scope can be controlled by a cleaver use of fresh ``registers'', i.e., unary predicate symbols in $\tau_{\rm reg}$.}
	
\end{enumerate}
Thus, the semantics associates, with each algebraic expression $\term$, a functional binary relation (i.e., a partial mapping) $\che(\term)$ on strings of relational structures, that depends on  a specific Choice function $\ch$.

\vspace{1ex}

\begin{example}\label{ex:strings-as-functions}
	{\rm
		Consider term $\module_1 \comp \module_2(\e)$ and structure $\strA$. 	Suppose  $\ch$ is such that 
		\begin{align}
			(\strA \mapsto   \strA\cdot \strB) & \in 	\ch (\module_1(\e)), \nonumber \\ 
			(\strA \cdot \strB  \mapsto  \strA\cdot \strB\cdot \strC) & \in 	\ch (\module_2(\e)). \nonumber
		\end{align}
		Then
		$
		(\module_1 (\e)\comp \module_2(\e) ) ^{\Tr}(\ch/\e)(\emp) = \strA \cdot \strB \cdot \strC.
		$
		The same word can also be obtained using infinitely many  different terms, e.g.,
		$
		(\module_1\comp  \rneg \rneg \id \comp \module_2\comp \id (\e)) ^{\Tr}(\ch/\e) (\emp) = \strA \cdot \strB \cdot \strC
		$, and, similarly $( \rneg \rneg  \rneg \rneg \module_1 \comp  \module_1\comp  \rneg \rneg \module_2 \comp \module_2(\e)) ^{\Tr}(\ch/\e) (\emp) = \strA \cdot \strB \cdot \strC$.
	}
\end{example}

\subsubsection{Maximum Iterate vs the Kleene Star}
To  give a comparison of Maximum Iterate with the Kleene star (i.e., the reflexive transitive closure, a construct known from Regular Languages), we show how to define the two operators side-by-side, inductively.

\colorlet{shadecolor}{gray!50}
\tikzset{My Style/.style={black, draw=shadecolor,fill=shadecolor, minimum size=0.1cm}}

\begin{center}

\begin{minipage}{.4\textwidth}
	\begin{figure}[H]
		\begin{center}
			\vspace{-3ex}
			\tikzset{every picture/.style={line width=0.75pt}} 

\begin{tikzpicture}[x=0.75pt,y=0.75pt,yscale=-.5,xscale=.5]

\draw   (213.74,93.03) .. controls (213.74,87.84) and (208.74,83.63) .. (202.57,83.63) .. controls (196.4,83.63) and (191.4,87.84) .. (191.4,93.03) .. controls (191.4,98.23) and (196.4,102.44) .. (202.57,102.44) .. controls (208.74,102.44) and (213.74,98.23) .. (213.74,93.03) -- cycle ;
\draw   (267.74,93.03) .. controls (267.74,87.84) and (262.74,83.63) .. (256.57,83.63) .. controls (250.4,83.63) and (245.4,87.84) .. (245.4,93.03) .. controls (245.4,98.23) and (250.4,102.44) .. (256.57,102.44) .. controls (262.74,102.44) and (267.74,98.23) .. (267.74,93.03) -- cycle ;
\draw   (160.74,93.03) .. controls (160.74,87.84) and (155.74,83.63) .. (149.57,83.63) .. controls (143.4,83.63) and (138.4,87.84) .. (138.4,93.03) .. controls (138.4,98.23) and (143.4,102.44) .. (149.57,102.44) .. controls (155.74,102.44) and (160.74,98.23) .. (160.74,93.03) -- cycle ;
\draw    (149.57,83.63) .. controls (189.57,53.63) and (255,5.69) .. (341.83,93.53) ;
\draw [shift={(341.83,93.53)}, rotate = 225.37] [color={rgb, 255:red, 0; green, 0; blue, 0 }  ][line width=0.75]    (17.64,-4.9) .. controls (13.66,-2.3) and (10.02,-0.67) .. (6.71,0) .. controls (10.02,0.67) and (13.66,2.3) .. (17.64,4.9)(10.93,-4.9) .. controls (6.95,-2.3) and (3.31,-0.67) .. (0,0) .. controls (3.31,0.67) and (6.95,2.3) .. (10.93,4.9)   ;
\draw    (202.57,83.63) .. controls (280.33,32.17) and (291.67,52.5) .. (331.6,83.9) ;
\draw    (256.57,83.63) .. controls (288,61.69) and (307,61.17) .. (331.6,83.9) ;
\draw [line width=3]    (47.4,93.03) -- (138.4,93.03) [My Style] ;
\draw [line width=3]    (160.74,93.03) -- (191.4,93.03)[My Style]  ;
\draw [line width=3]    (213.74,93.03) -- (245.4,93.03)[My Style]  ;

\draw (139,108.8) node [anchor=north west][inner sep=0.75pt]   [align=left] {$\strA_1$};
\draw (193,109.2) node [anchor=north west][inner sep=0.75pt]   [align=left] {$\strA_2$};
\draw (246,108.6) node [anchor=north west][inner sep=0.75pt]   [align=left] {$\strA_3$};
\draw (325.4,107.8) node [anchor=north west][inner sep=0.75pt]   [align=left] {$\strB$};

\end{tikzpicture}
		\end{center}
		\caption{Maximum Iterate: \\ $\term^\iter_0:= \rneg t$,\ \ \  $\term^\iter_{n+1}:= \term^\iter_{n} \comp  t$, \ \ \ $\term^\iter := \bigcup_{n\in\nat} \term^\iter_n$. \\ It is a deterministic operator (a partial function). }
		\label{fig:Kleene-star}
	\end{figure}
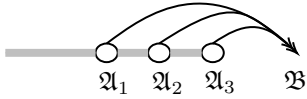
\end{minipage}
{\ \ \ }
\begin{minipage}{.5\textwidth}
	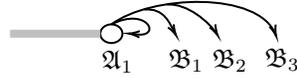
\begin{figure}[H]
		\begin{center}
			\tikzset{every picture/.style={line width=0.75pt}} 

\begin{tikzpicture}[x=0.75pt,y=0.75pt,yscale=-0.5,xscale=0.5]

\draw   (221,61.93) .. controls (221,56.74) and (226,52.53) .. (232.17,52.53) .. controls (238.33,52.53) and (243.33,56.74) .. (243.33,61.93) .. controls (243.33,67.12) and (238.33,71.33) .. (232.17,71.33) .. controls (226,71.33) and (221,67.12) .. (221,61.93) -- cycle ;
\draw [line width=3]    (130,61.93) -- (221,61.93)[My Style]  ;
\draw    (232.17,52.53) .. controls (271.77,22.83) and (380.46,18.36) .. (397.84,62.77) ;
\draw [shift={(398.33,64.14)}, rotate = 251.14] [color={rgb, 255:red, 0; green, 0; blue, 0 }  ][line width=0.75]    (10.93,-3.29) .. controls (6.95,-1.4) and (3.31,-0.3) .. (0,0) .. controls (3.31,0.3) and (6.95,1.4) .. (10.93,3.29)   ;
\draw    (265.43,37.63) .. controls (285.38,35.73) and (290.89,50.98) .. (296.53,61.94) ;
\draw [shift={(297.43,63.63)}, rotate = 241.39] [color={rgb, 255:red, 0; green, 0; blue, 0 }  ][line width=0.75]    (10.93,-3.29) .. controls (6.95,-1.4) and (3.31,-0.3) .. (0,0) .. controls (3.31,0.3) and (6.95,1.4) .. (10.93,3.29)   ;
\draw    (288.43,32.63) .. controls (308.48,30.72) and (331.27,50.7) .. (338.51,62.09) ;
\draw [shift={(339.43,63.63)}, rotate = 241.39] [color={rgb, 255:red, 0; green, 0; blue, 0 }  ][line width=0.75]    (10.93,-3.29) .. controls (6.95,-1.4) and (3.31,-0.3) .. (0,0) .. controls (3.31,0.3) and (6.95,1.4) .. (10.93,3.29)   ;

\draw    (232.17,52.53) .. controls (269.94,33.45) and (284.14,68.04) .. (255,62.27) ;
\draw [shift={(249.7,61.77)}, rotate = 363.92] [color={rgb, 255:red, 0; green, 0; blue, 0 }  ][line width=0.75]    (10.93,-3.29) .. controls (6.95,-1.4) and (3.31,-0.3) .. (0,0) .. controls (3.31,0.3) and (6.95,1.4) .. (10.93,3.29)   ;

\draw (220,75.8) node [anchor=north west][inner sep=0.75pt]   [align=left] {$\strA_1$};
\draw (287,74.67) node [anchor=north west][inner sep=0.75pt]   [align=left] {$\strB_1$};
\draw (331,74.67) node [anchor=north west][inner sep=0.75pt]   [align=left] {$\strB_2$};
\draw (384,73.67) node [anchor=north west][inner sep=0.75pt]   [align=left] {$\strB_3$};

\end{tikzpicture}
		\end{center}
		\caption{The Kleene Star:\\
			 \ \, $\term^*_0:= \id$,\  \  \  $\term^*_{n+1}:= \term \comp \term^*_{n}$, \  \ \ $\term^* := \bigcup_{n\in\nat} \term^*_n$. \\
			  It is a non-deterministic operator (not a function).}
		\label{fig:maximum-iterate}
	\end{figure}
\end{minipage}

\end{center}

		\subsection{Terms as Partial  Functions (String-to-String Transductions)}

				\vspace{1ex} 
		
	The  given semantics  allows us to view algebraic terms, parameterized by Choice functions, as (partial) functions.  Let us  make the functional view more explicit. First, we do it for atomic modules, and then generalize to all terms.
		
			\vspace{1ex} 

	\noindent\textbf{Atomic Modules} Given a specific Choice function, each module symbol $\module$  is interpreted as a semantic partial  mapping $\module^{\Tr}$ on strings in $\Un^+$:
	\begin{equation}\label{eq:atomic-mapping}
		\module^{\Tr}(\ch/\e)(\sstring_i ) = \sstring_{i+1}    \ \text{ iff } \   (\sstring_i \mapsto \sstring_{i+1}) \in \ch(\module(\e)),
	\end{equation}
	where, as before,  a semantic instantiation of the function variable $\e$ by the concrete Choice  function  $\ch$ is denoted $(\ch/\e)$.

	\vspace{3ex} 
	
	\noindent\textbf{All Terms}  	
	The partial mappings (\ref{eq:atomic-mapping}) are extended from atomic to all terms, as presented in Section \ref{sec:extension-of-Choice-functions}.	Each term, with a specific Choice function, is interpreted as  a partial function $\term^{\Tr}(\ch/\e):  \Un^+ \partialmap \Un^+$ as follows:
	\begin{equation}\label{eq:mapping}
		\begin{array}{l}
			\term^{\Tr}(\ch/\e)(s_i) = s_j      \ \text{ iff } \  \text{ the extension $\che$  of $\ch$ }  
			\text{for $\term$ exists, }
			\text{ and } (\sstring_i\mapsto \sstring_j) \in \che(\term).
		\end{array}
	\end{equation}

Recall that the terms  of the algebra take other functions as arguments.  
The summary of the syntax is presented in the table  in Section \ref{sec:Algebra}. In that table, $\cF$ denotes	(partial) functions on $\Un^+$.
Thus, according to (\ref{eq:mapping}), the partial functions  in the closure of ${\cal F}$ under all algebraic operations (\ref{eq:algebra}) map strings of structures to strings of structures, i.e., are \emph{string-to-string transductions}.   
This \emph{functional} view is essential to develop an  algebraic formalization of  non-deterministic computations.

	\subsection{Tests vs Processes}\label{sec:tests-vs-processes}

	Tests play a special role in our formalisation of non-deterministic computations as they specify computational decision problems.
	While tests are, essentially, ``yes'' or ``no'' questions, processes specify the actual content of these questions. We now define these notions formally.
	
	\begin{definition}\label{def:test-process}
	A term $\term$ is called a \emph{test} if, for all $\ch$, the extension $\che$ of $\ch$ returns an identity mapping, that is, if in the interpretation $\che(\term)$ of $\term$,  we have that $i=j$ in each map $(\sstring_i \mapsto \sstring_j ) \in  \che(\term)$.  Otherwise, $\term$ is called a \emph{process}. 
\end{definition}
\noindent Thus, every test is a subset of Identity (Diagonal) relation on $\Un^+$.  The cases for tests in the definition of an extension of Choice function are 2,3,7 and 8.

\vspace{1ex}

Note that it is possible that an atomic module does not change the interpretations of any relations (``registers''), and thus does not make a transition to another state in the transition system $\Tr$. However, such a module would not be considered a test because it still extends the current string by repeating the same letter. 

\vspace{1ex}

Notice an important difference between  tests and processes.  A test can ``reject'' an input string because the domain of the corresponding partial function is empty. For example, 
 $\rneg \id(x)$ rejects all strings because $\che(\rneg \id)$ is  always undefined. A test can also ``accept'' a string. 
For instance,  $\id(x)$ is everywhere defined and thus ``accepts'' every string. The situation with processes is quite different. Every process defines a string-to-string transduction. It extends a string, provided it is defined in that string.

\section{Dynamic Logic}\label{sec:Dynamic-Logic}

  We now
provide an alternative (and equivalent) two-sorted version of the syntax  
in the form of  a  Dynamic logic. 
While the two formalizations are equivalent, in many ways, it is easier to work with the logic. 
The syntax is given by the grammar:  
\begin{equation}
\label{eq:dynamic-logic}
\begin{array}{c}
\hspace{-4ex}	\term  :: =    \module(\e) \mid  
\!	\id \! \mid  \rneg\term   \mid  
	\term\comp \term  \mid 
	\term  \sqcup \term \mid     \term ^\iter   
		 \mid\! P= Q \!\mid \!\BG(P\neq Q)\!  \mid \!  \phi ? \\ 
\phi :: = \   | \term \rangle \phi \mid  \mT 
 \ \mid \ \neg  \phi \  \mid \ \phi \land \phi \ \mid \ \phi \lor \phi 
\end{array}
\end{equation}
In Dynamic Logic literature, the expressions in the first line 
are typically called \emph{process} terms, and those in the second line \emph{state} formulae or \emph{tests} (cf. Definition \ref{def:test-process}). 
	State formulae $\phi$  are ``unary'' in the same sense as $P(x)$ is a unary notation for $P(x,x)$. Semantically, they are subsets of the identity relation. 
The  state formulae in the second line of (\ref{eq:dynamic-logic}) are 
 shorthands that use the operations in the first line: $  | \term \rangle \ \phi := \rneg \rneg (\term \comp \phi)$ is the ``diamond'' modality claiming the existence of a  successful execution of $\term$ leading to $\phi$. In addition, we define more  familiar  operations $ \mT : = \id$, 	 
 $ \neg  \phi : = \rneg \phi$, \ \  $\phi \land \psi := \phi \comp \psi$, \ \ $\phi \lor \psi := \phi \sqcup \psi$, \  and  also $ \phi ? :=  \rneg \rneg  \phi$, appearing in the top line of (\ref{eq:dynamic-logic}). 
   In addition, we define $\Dom(\term) := \rneg \rneg (\term)$.

An interesting property is  that Unary Negations are not idempotant (double negation does not cancel to a negation-free expression), 
   but, they  are weakly idempotant: 
   a triple negation of it cancels to a single one.
  We will come back to this property in Examples \ref{ex:three-negations} and \ref{ex:two-negations}, after the notion of a strong term equality is introduced in Section \ref{sec:Strong-Equalities}.

 \vspace{1ex}

 As a  general rule, we will refer to algebraic process-terms of the form $\phi?$  in the first line of (\ref{eq:dynamic-logic}) as tests, and to expressions in the second line of (\ref{eq:dynamic-logic}) as state formulae. But,  the main difference between the terms ``test'' and ``state formulae'' is mainly in the context of their use, since state formulae are merely a syntactic sugar of the algebra. Clearly, any process, including $\phi?$, can be transformed into a state formula by adding Antidomain ($\rneg\!$ ) or Domain ($\rneg\rneg$) to it.

\subsection{Semantics of Dynamic Logic } To interpret the Dynamic Logic (\ref{eq:dynamic-logic}), we need to provide semantics to both processes and state formulae. Processes 
are interpreted as before -- as  partial mappings (parameterized with $\ch$) on strings in $\Un^+$, that can be viewed as string-to-string transductions (\ref{eq:mapping}). For state formulas we define, for all concrete Choice functions $\ch$ and strings $\sstring$,
\begin{equation}	\label{eq:state-formulas}
\Tr, \sstring \models \phi(\ch/\e) \ \ \  \text{ iff }  \ \ \  \phi^{\Tr}(\ch/\e)(s) = s.
\end{equation}
Notice that (\ref{eq:state-formulas}) is a particular case of  (\ref{eq:mapping}), the partial mapping specified by a term, where the terms are tests (cf. Definition \ref{def:test-process}). This is immediate from  cases  2,3,7 and 8 of the extended Choice function $\che$. 

In this paper, the relation $\Tr$ (that is specified by a secondary logic at the bottom level of the algebra) is fixed, and is omitted from  $\Tr, \sstring \models \phi(\ch/\e)$ (cf. (\ref{eq:state-formulas})), where $\phi$ is a state formula. But, of course, other settings are possible as well.

\subsection{Implicit Quantification in State Formulae }  \label{sec:implicit-quantification}
State formulae exhibit implicit quantification over the Choice function variable $\e$.
Indeed, tests of the form $\rneg\rneg\term$ represent the domain of $\term$ (which is, semantically, is a partial mapping (\ref{eq:mapping})), and, intuitively, mean that there is a Choice function $\ch$ such that 
 $ \sstring \models \rneg\rneg\term(\ch/\e)$.  On the other hand, tests of the form  $\rneg\term$ are universal, because they claim that there is no Choice function that corresponds to a successful execution of $\term$. To summarize, we have the following implicit quantification over Choice functions: 
 \begin{equation}\label{eq:quantifiers}
 		\begin{array}{rcl}
 		\rneg\rneg\term(\e)(x)  & -   &\text{ implicitly, } \exists \e, \\
 		\rneg\term(\e)(x)  & - &  \text{ implicitly, } \forall \e.
 	\end{array}
 \end{equation}
These are \emph{second-order} quantifiers, since we quantify over functions. The quantifiers can interleave, since expressions of the form (\ref{eq:quantifiers})  may appear within the term $\term$. 
The two kinds of tests  mentioned in (\ref{eq:quantifiers}) are of special interest in the study of computational problems specified by a term, as  discussed in Section \ref{sec:term-computational-problem}.

\vspace{1ex}

\subsection{Programming Constructs} \label{sec:programming-constructs} It is well-known that in Propositional Dynamic Logic  \cite{DBLP:journals/jcss/FischerL79}, imperative programming  constructs  are definable using a fragment of  regular languages, see the Dynamic Logic book by Harel, Kozen and Tiuryn \cite{HarelKozenTiuryn-DL}. The corresponding  language is called \emph{Deterministic Regular (While) Programs} in  \cite{HarelKozenTiuryn-DL}.\footnote{	Please note that Deterministic Regular expressions and the corresponding Glushkov automata are unrelated to what we study here. In those terms, expression $a \comp a^*$ is Deterministic Regular, while $a^* \comp a$ is not. Both expressions are not in our language.}  In our case, imperative  constructs  are definable (cf. \cite{JacksonStokes:2011}) by: 
$$
\begin{array}{ll}
	\begin{array}{l}
		{\bf skip}\ : =\ \id,  \ \ \ \ 
		{\bf fail}\ : =\ \rneg \id,\\
		{\bf if} \ \phi \ {\bf then}\  \term \ {\bf else} \ \term' \ : = \ (\phi? \comp \term) \sqcup   \term',
	\end{array} 
	&
	\begin{array}{l}
		{\bf while} \ \phi \ {\bf do}\  \term : = (\phi? \comp \term)^\iter  \comp (\rneg \phi?),\\
		
		{\bf repeat} \ \term \ {\bf until}\  \phi : = \term \comp ((\rneg \phi?) \comp \term)^\iter  \comp  \phi?.
	\end{array} 	
\end{array} 	
$$
\noindent Thus,  importantly, the non-determinism  of  operations $*$ and $\cup$ of regular languages is not needed to formalize these programming constructs.


\subsection{Duality between Sets of Strings and Transductions}
Recall the difference between tests and processes discussed in Section \ref{sec:tests-vs-processes}.
Notice that, an application of a process-term $\term$ can only make a string longer.  
This is because,  in (\ref{eq:mapping}), if the extension $\che(\term)$ of $\ch$ for process $\term$ exists, it is always the case that, in $(\sstring_i \mapsto \sstring_j ) \in  \che(\term)$, we have $i < j$. 
Thus, the mapping associated with a process-term can be viewed as a set of (non-empty) strings over the alphabet $\Un$. We call the (non-empty) strings in this set the \emph{deltas of $\term$}. In certain contexts, we also call them  \emph{traces} of $\term$. The term ``delta'' is more appropriate when we talk about applying a term, e.g., when we say that we \emph{apply a  string delta selected by $\che$}. A ``trace'' is something left \emph{after} the execution of $\term$ from an inputs structure $\strA$. We talk more about traces in Section (\ref{sec:traces}). Notice that \emph{tests} never extend a string. Their deltas are the empty strings.

\colorlet{shadecolor}{gray!50}
\tikzset{My Style/.style={black, draw=shadecolor,fill=shadecolor, minimum size=0.1cm}}

\begin{figure}[H]
	\begin{center}
		\input{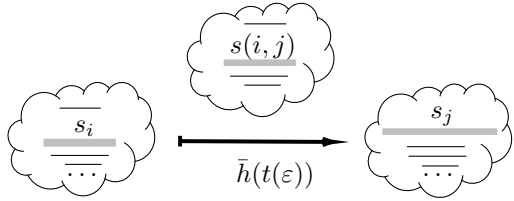}
	\end{center}
	\caption{ The set of strings in the middle, shown above the arrow,  is the set of deltas of $\term(\e)$. For $(\sstring_i \mapsto \sstring_j ) \in  \che(\term(\e))$, the string $\sstring(i,j)$ is one of such deltas. It extends $\sstring_i$ and produces  $\sstring_j$. Note that, even if the set of deltas  is finite,  the domain and codomain of $\che(\term(\e))$ are, in general, infinite (shown on the left and on the right). }
	\label{fig:cloud}
\end{figure}
In fact, deltas of the form $\sstring(1,j)$, for some $j$, are associated with certificates for computational decision problems, and are related to some equivalence classes of Choice functions that ``agree'' on $\sstring(1,j)$, as will be seen shortly, in Section \ref{sec:traces}.

In summary,   the interpretation of each process-term can be viewed as a set of deltas -- non-empty strings -- ``extensions'' applied to strings on the input.

\label{sec:structure-choice-functions}

\section{Main Computational  Task} 
\label{sec:main-task}

In this section, we formulate the main computational task we study in this paper.   The task amounts to satisfiability, with respect to a one-letter string,
of a state formula in certain form. That is, the task is a particular case of  satisfying a state formula, cf. (\ref{eq:state-formulas}).
Importantly, we define  a computational problem specified by an algebraic term, as a class of relational structures. We then discuss how certain equivalence classes of Choice functions, that, intuitively,  correspond to traces of computation, can serve as certificates for such computational problems.

Recall that, in this paper, relation $\Tr$ is fixed, and is omitted from  $\Tr, \sstring \models \phi(\ch/\e)$  in (\ref{eq:state-formulas}). This state formula    $\phi$ will be of a special form, in the computational task we define next.

\vspace{3ex}

\noindent \fbox{\parbox{\dimexpr\linewidth-2\fboxsep-2\fboxrule\relax}{ Problem: \textbf{Main Task (Decision Version)} 
		
		\underline{Given:} A structure $\strA$ with an empty vocabulary and term $\term$. \\
		\underline{Question:} 
		\begin{equation}\label{eq:main-task}
			\exists \ch \	\strA \models | \term \rangle  \mT(\ch/\e)  \ ?
\end{equation}}}

\vspace{1ex} 

\noindent Intuitively,  when the answer is yes, (\ref{eq:main-task})
says that there is a successful execution of $\term$ at the input structure $\strA$, e.g., a graph, and  $\ch$ is a witness of it.

The study of the decision version of the computational task (\ref{eq:main-task}) is the main goal of this paper.
In particular, we are interested in the data complexity of this task, where the 
formula is fixed, and the input structures vary \cite{Vardi82}.\footnote{Notice that  fixing the relation $\Tr$ means fixing a \emph{formula} of the secondary logic. The size of the transition graph, i.e., the interpretation of the specification in the secondary logic,  depends on the size of the input structure.}

\vspace{1ex}

 Dually to this task, the complement of the problem  (\ref{eq:main-task})  can be defined, using the fact that $ \rneg \term $ is the complement of $ \rneg \rneg \term $, which is a syntactic variant of 
 $ | \term \rangle  \mT$. For instance, if a problem in the complexity class NP is specified using the, implicitly, existential, term $ \rneg \rneg \term $, then its complement, a problem in co-NP, is specified by the term $ \rneg \term$, with an implicit universal quantifier. The reader who is familiar with some elements of Descriptive Complexity (see, e.g., Immerman's book \cite{Immerman-book}), might notice 
 an analogy with the connection  between second-order quantifier alternations and the polynomial time hierarchy. 
 
 \vspace{1ex}

\noindent\textbf{Search Version} 
 The state formula $| \term \rangle  \Last(\e)$ may be seen as \emph{defining a set of Choice functions} with respect to the inputs structure $\strA$. 
 This gives us a \emph{search} version of the main task, which tells us to find all such  $\ch$. Recall that only one  free variable $\e$ is allowed per term.\footnote{There is an analogy with evaluating 
 	queries in database theory. 
 	A	formula $\phi(\bx)$ with free variables $\bx$ in classical logic is viewed as a query to a database (relational structure).  The query $\strA \models \phi(\bx)$ returns a set of tuples of domain element that, when instantiated for the free variables  $\bx$, make the formula true in the structure.}

\subsection{Computational Problem $\Pa$ Specified by Algebraic Term $\term$}\label{sec:term-computational-problem}

\begin{definition}\label{def:comp-problem}
	A \emph{computational problem  specified by a process-term $\term$} is an isomorphism-closed  class $\Pa$ of $\tau$-structures $\strA$ such that a structure $\strA$  is in this class if and only if 
	(\ref{eq:main-task}) holds, that is, there exists 
	$ \ch$ such that  
	$
	\strA \ \models \  |\term\rangle \Last\ (\ch/\e).
	$\footnote{For natural secondary logics such as the one given in Section \ref{sec:Atomic-Modules}, isomorphism-closure holds for all structures satisfying the Main Task \ref{eq:main-task}.}
\end{definition}
\noindent Here, we assume that, in $\strA$, all relational ``register'' symbols in $\tau_{\rm reg}$ are interpreted by a special ``blank'' element, and the interpretation  of $\tau_{\rm EDB}$  describes the input to the problem.

Intuitively, the class $\Pa$  contains all structures $\strA$ such that  a successful execution of $\term$ on input $\strA$ is possible, or $\term$ is \emph{defined} on input $\strA$.  We will talk more about defined and undefined terms in Section \ref{sec:Strong-Equalities}.

\vspace{1ex}

Many Choice functions witness the main task (\ref{eq:main-task}) in the ``same way''. We now explain this notion of similarity formally, by introducing equivalence classes on Choice function.

\subsection{Traces as Equivalence Classes of Choice Functions} \label{sec:traces}  

We have seen, in (\ref{eq:main-task}), that Choice functions witness the main computational task. 
But, we don't want to distinguish between  Choice functions that are, in some sense, ``similar'', as certificates. This similarity relates to the notion of a trace.
We say that string $\sstring$ is  a \emph{trace of term $\term$} from input $\strA$  if there is $\ch$ such that 
$(\strA \mapsto \sstring) \in \che(\term(\e))$. An example of a trace of  $\term \ := \ \module_{17}\comp \module_{17}\comp\module_{17}(\e)$  is highlighted in Figure \ref{fig:tree} by the shaded area.

For the same Choice function, many different terms can produce the same trace, say $\strA \cdot \strB \cdot \strC$, as Example \ref{ex:strings-as-functions} in Section \ref{sec:extension-of-Choice-functions}
 shows.
Dually, for the same term $\term$ and an input structure $\strA$, different Choice functions can generate the same trace, but act differently outside of it.
We consider two different Choice functions $\ch$ and $\ch'$ \emph{equivalent with respect to $\term$ and $\strA$,} if for that term,  they agree on a trace $\sstring$ of $\term$ from $\strA$, i.e., return the same mapping $\strA\mapsto\sstring$:
\begin{equation}\label{eq:equiv-relation}
	\ch \cong_{(\strA,\term)} \ch'  \ \ \Leftrightarrow \ \  (\strA\mapsto \sstring) \in \che(\term)  \mbox{ iff } (\strA\mapsto \sstring)\in \bar{\ch'}(\term), \mbox{ for some string } \sstring\in\Un^+.
\end{equation}
This equivalence relation gives us equivalence classes $[\ch]_{(\strA,\term)}$  of Choice functions.	
Each equivalence class has a one-to-one correspondence with  a trace   of $\term$ from $\strA$, so we can refer to traces as equivalence classes $[\ch]_{(\strA,\term)}$ (i.e., with a string $\sstring$ on which the Choice functions agree). 
We will see shortly that these equivalence classes (or traces) can be viewed as certificates for the membership of $\strA$ in  a computational problem specified by a term.

\begin{example}
	{\rm Consider a tree that represents the unwinding of the transition relation $\Tr$ from the input structure $\strA$ in Figure \ref{fig:tree}. For each branch from  $\strA$, that is a trace of some term $\term$, there could be  Choice functions that agree on that branch, but act differently (return different transductions for the same term $\term \ := \ \module_{17}\comp \module_{17}\comp\module_{17}(\e)$)  on nodes not in that branch.  
For example,  some other Choice function $g$ may agree with $\ch$ on the trace highlighted by the shaded area in the figure, but can return different mappings  for 
the same term, e.g., 
for the node $\strB$, which is not in that trace. So,  $\ch \cong_{(\strA,\term)} g$, but $\ch \not \cong_{(\strB,\term)} g$.
} 
\end{example}

\begin{example}
	{\rm For all atomic modules $\module \in \Modules$ and all structures $\strA \in \Un$, we have
$$
[\ch]_{(\strA,\module)} \ :=\ \{\ch \mid  \exists \strB\  (\strA,\strB) \in \ch(\module(\e))\}.
$$
} 
\end{example}

\begin{example}
	{\rm Interesting, Identity, a term counterpart of a tautology, generates only ``realistic'' promises, i.e., those on which at least one module is defined:
		$$
		[\ch]_{(\strA,\id)} \ :=\ \{\ch \mid \exists \module \in \Modules \text{ and } \exists \strB\  (\strA,\strB) \in \ch(\module(\e))\}.
		$$
	} 
\end{example}

\vspace{2ex}

\subsection{Promises (a.k.a. Witnesses or Certificates)}

We define the \emph{set of witnesses for $\strA$ in  $\Pa$} as 
$$
W^{\term}_{\strA} \ : =\ \{ [\ch]_{(\strA,\term)} \mid    \strA \ \models \  |\term\rangle \Last\ (\ch/\e)        \} \ \ \ \ \text{ and }\ \ \  W^{\term} \ : =\ \ \bigcup_{\strA \in \Pa} W^{\term}_{\strA}.
$$

We also call the witnesses ``promises'', which gives the name to the algebra (\ref{eq:algebra}) -- Promise Algebra. If a promise is given, the computation is guaranteed to succeed.

\subsection{Boolean Algebra of Promises}\label{sec:boolean-algebra-promises}

The intended use of the algebra is to study queries of the form (\ref{eq:main-task}). In that use, strings in $\Un^+$ are partitioned into those that are traces of some algebraic terms that represent programs, 
	and thus can potentially act as ``yes''-certificates to computational decision problems, and those that are not traces of any term whatsoever, and thus cannot be ``yes''-certificates for any $\Pa$. 
We can consider the Boolean algebra of the set of all potential ``yes'' certificates as follows. 

	\vspace{1ex}
	
	The Boolean algebra $B(A)$ of a set $A$ is the set of subsets of $A$ that can be obtained by means of a finite number of the set operations union (OR), intersection (AND), and complementation (NOT), see pages 185-186 of Comtet's 1974 book \cite{Comtet-book74}.

	\vspace{2ex}

 Let $\cH$ be the set of all promises, i.e., equivalence classes on functions from $\ChFunctions$, with respect to the equivalence relation (\ref{eq:equiv-relation}).
 
 \begin{definition}\label{def:Boolean-algebra-of-promises}
 The \emph{Boolean Algebra of Promises} is the  Boolean algebra  of the set 
$
\{(\sstring \mapsto \sstring) \mid \sstring \in \cH \}.
$
\end{definition}
\noindent For simplicity, since pairs $(\sstring \mapsto \sstring)$ always contain identical strings in $\cH$, we denote this Boolean algebra as $B(\cH)$.

\vspace{1ex}

If $ \sigma$ is a signature and   $A$, $B$ are  $\sigma$-structures (also called $ \sigma$-algebras in
 universal algebra), then a map 
 $h:A\to B$  is a \emph{$\sigma$-embedding} if all of the following hold:
\begin{itemize}
	\item $h$ is injective,
\item for every $n$-ary function symbol $f\in \sigma$ and  $a_{1},\ldots ,a_{n}\in A^{n}$, we have 
\begin{equation}\label{eq:embedding-mapping}
h(f^{A}(a_{1},\ldots ,a_{n}))=f^{B}(h(a_{1}),\ldots ,h(a_{n})).
\end{equation}
\end{itemize}
 The fact that a map $h:A\to B$ is an embedding is indicated by the use of a ``hooked arrow''  $h:A \hookrightarrow B$.\footnote {Note that this
 notation is sometimes also used for inclusion maps, but, here, we mean an embedding (which happens to be an inclusion).}

\begin{proposition}\label{prop:boolean-algebra}
Let $\sigma = \{\land,\lor,\neg\}$. There is a $\sigma$-embedding of the the Boolean algebra $B(\cH)$ into the Dynamic Logic  (\ref{eq:dynamic-logic}).

\end{proposition}

\begin{proof} 
		We  need a structure-preserving injective mapping $h:A \hookrightarrow B$, where 
	$A$ is the Boolean Algebra of Promises, $B(\cH)$, and 
$B$ is the Dynamic Logic. The elements $a_i \in A$ in (\ref{eq:embedding-mapping}) are, in our case, sets of partial  mappings of the form $(s\mapsto s)$.
		The mapping  $h:A \hookrightarrow B$ is such that all sets of maps $(s\mapsto s)$ of promises are mapped to themselves. The homomorphism property (\ref{eq:embedding-mapping}) clearly holds because, for 
$S, T \subseteq 	\{(\sstring \mapsto \sstring) \mid \sstring \in  \cH\}$, we have 
  $h(S\land^A T) = h(S)\land^B h(T)$, and similarly for 
	$\lor$ and $\neg$, by the semantics of state formulae (\ref{eq:state-formulas}). \end{proof}

The following proposition is straightforward, as it follows immediately from the semantics of terms as string-to-string transductions and the definition of the equivalence relation (\ref{eq:equiv-relation}).

\begin{proposition}
	The underlying set $\cH$ of the  Boolean algebra $B(\cH)$ has a forest structure, with a tree for each term.
\end{proposition}

\section{Strong Equality for Reasoning about Computation}
 \label{sec:Strong-Equalities}
In this section, we explain that the truth of certain (conditional) equalities between terms indicates the existence, or non-existence, of a ``yes'' certificate for a computational problem specified by a term.

Recall that terms are interpreted as \emph{partial} functions, see (\ref{eq:terms-partial-maps}).  This is crucial for reasoning about computations (cf. the Main Task (\ref{eq:main-task})), since programs are not everywhere defined. 
\subsection{Defined and Undefined Terms}
\begin{definition}\label{def:defined-undefined}
	We say that term $\term$ is \emph{defined in $\sstring_i$}, notation $\term(\sstring_i)\!\dn$, if there is $\ch\in \ChFunctions$ and  string $\sstring_j$ such that  $\term^{\Tr}(\ch/\e)(s_i) = s_j $, otherwise $\term$ is \emph{undefined in $\sstring_i$}, notation $\term(\sstring_i)\notdn$.  
\end{definition}
\noindent Intuitively, definedness is associated with the existence of a ``yes'' certificate, and undefinedness  is associated with the non-existence of a ``yes'' certificate of a computational problem.  	Thus, another way of viewing the computational problem $\Pa$ specified by term $\term$  is as an isomorphism-closed class of structures $\strA$ such that  $\term(\strA)\dn$.

\vspace{2ex}

\subsection{Trace Equivalence (a.k.a. Strong Equivalence)}

Let  the state-to-state transition relation $\Tr$ (given by the secondary logic), be fixed. Let $\sstring$ be a   string in  $\Un^+$.

\begin{definition}

Terms $\term$ and $g$ are
	\emph{strongly equivalent (or trace-equivalent) on string $\sstring\in \Un^+$}, notation $\term(\sstring)\traceeq  \termg(\sstring)$, if 
		\begin{enumerate}
		\item 
		they both are defined on $\sstring$, in symbols, $\term(\sstring)\dn$ and   $\termg(\sstring)\dn$, and 
		\item for all $\ch\in \ChFunctions$,  they denote the same mapping on $\sstring$, i.e., 
$
		\term^{\Tr}(\ch/\e)(\sstring) =  g^{\Tr}(\ch/\e)(\sstring).
$
		\end{enumerate}
We say that terms $\term$ and $g$ are
	\emph{strongly equal on a set $S \subseteq \Un^+$} if  for all strings $\sstring$ in $S$, we have that $ \term(\sstring)\traceeq  \termg(\sstring)$.
\end{definition}
	Notice that terms that are trace-equivalent on some set $S$  have the same  sets of deltas, extending strings in that set $S$.

\begin{example}
{\rm
	All tests that are vacuously true on all strings $\sstring\in \Un^+$ (are tautologies) are strongly equal to Identity $\id$ on the set of all strings $\Un^+$.
}
\end{example}	
\begin{example}\label{ex:three-negations}
{\rm Strong equivalence  $ \rneg \rneg \rneg \term(\sstring) \traceeq \rneg \term(\sstring)$ holds on the set of those strings where $\term$ is undefined (equivalently,  $\rneg \term$ is defined). The equivalence holds  because Anti-Domain of Domain (i.e., $ \rneg \rneg \term(\sstring)$) is Anti-Domain.  
}
\end{example}	
	\begin{example}\label{ex:two-negations}
		{\rm On the other hand, if $\term$ is a process-term, then there is no set $S$ on which  $ \rneg \rneg \term(\sstring) \traceeq  \term(\sstring)$. We leave it to the reader to figure out why this is the case (one has to construct a counterexample).
}
\end{example}	

Note that undefined terms are not considered (strongly) equal. Thus, \emph{strong equality is not reflexive}. This is an important property for reasoning about computation, as discussed shortly in Section \ref{sec:self-equality-definedness}.

		\vspace{2ex}

	\subsection{Before-After Equivalence  (a.k.a. Input-Output Equivalence)}
	
	To define the notion of a  Before-After equivalence of terms, we need the following notion of a projection onto $\Un$, which is simply the set of  all pairs (start, end) of  a term's deltas.
	
	A \emph{projection $\Proj$ of  term $\term$  under Choice function $\ch$ onto $\Un$} 
	is defined as:
	$$
	(\, \sstring(i),\sstring(j) \, ) \in \Proj(\term, \ch) \ \ \ \Leftrightarrow   \ \ \  (\sstring_i \mapsto \sstring_j ) \in  \che(\term).
	$$

	\begin{definition}\label{def:strong-equality}
		Terms $\term$ and $g$ are \emph{Before-After equivalent on $\sstring$}, notation $\term(\sstring)\ioeq \termg(\sstring)$, if 
		\begin{enumerate}
			\item  they both are defined on $\sstring$, in symbols, $\term(\sstring)\dn$ and   $\termg(\sstring)\dn$, and
			\item  for all $\ch\in\ChFunctions$, their projections onto $\Un$ coincide, i.e., 
	$
	\Proj(\term, \ch) = 	\Proj(\termg, \ch) .
 $		
		\end{enumerate}
	\end{definition}

	\vspace{2ex}

	Both equalities we have just introduced can be uses to specify the existence of a certificate, that invisible  elephant in the room, as we show next.
	
	\subsection{Strong Equality and Existence of a Certificate}\label{sec:self-equality-definedness}
	
	An important consequence of the definitions of Strong equality and Before-After equality is that, in both cases,  such an equality  is not necessarily reflexive. As a consequence,  
	\emph{self-equality  is identified with definedness}.  
	On the other hand, undefinedness, in $x$, denoted  $\term(x)\notdn$, is associated with  $\rneg \term(x)  \strongeq x$ (here, $\ioeq$ can be considered as well, depending on  specific application needs).
	
	 To summarize,
	$$
	\begin{array}{llll}
		\term(x)\! \dn &  \text{ abbreviates } &\term(x) \strongeq \term(x),&   \text{ a ``yes'' certificate for $x$ in $\Pa$ exists}, \\
		\term(x)\notdn & \text{ abbreviates } & \rneg \term(x) \strongeq  x, &  \text{ no ``yes'' certificates  for $x$  in $\Pa$  exist} .
	\end{array} 
	$$

		Since terms denote \emph{partial} mappings,  strong equalities such  as  $\term(x)   \strongeq  \termg(x)$ may hold for some strings $x$, but not for all $x$, as in Example \ref{ex:three-negations}. For that reason, the algebra (\ref{eq:algebra}) cannot be axiomatized by universally quantified equalities between terms. 
	Instead of equational theories that axiomatize algebraic structures such as groups,  we need to consider \emph{quasi-equational} theories, i.e., 
	those with \emph{conditional} equations in the Horn form. 
	
	The fact that the existence of a certificate is associated with equality of certain terms opens up a possibility  developing a proof system where \emph{deriving} the existence of a certificate is possible. Developing such a quasi-equational theory and a proof system is outside of the scope of this paper.

\section{ A Logic for Atomic Modules}\label{sec:Atomic-Modules}

In this section, first, we discuss the Law of Inertia, that must hold  regardless of what secondary logic is used for specifying atomic modules. Second, we define a specific secondary logic used further in this paper. The logic is based on a modification of Conjunctive Queries (CQs) -- only one element of the set generated by the CQ is, non-deterministically, returned on the output. Finally, we come back to the Law of Inertia for that specific secondary logic.

\subsection{Law of Inertia}

Recall that we have used the relation $ \semaNoh{\module(\e)}$ that represents a transition system in Definition \ref{def:Choice-function-atomic-modules} of the semantics of atomic modules, via Choice functions. We have, temporarily, treated it as an arbitrary binary relation. However, this relation has an important property that always holds for any update -- \emph{everything that is not explicitly modified in the update, must remain the same.} This property is commonly called the Law of Inertia,  starting from McCarthy and Hayes 1969  \cite{McCHay69}. 
Figure \ref{fig:Inertia} illustrates  the Law of Inertia for atomic modules. Recall that the vocabulary symbols $\tau$ are partitioned into EDB relations and ``registers'', $ \tau \ := \ \tau_{\rm EDB}   \uplus 
\tau_{\rm reg}$, see (\ref{eq:tau}).

	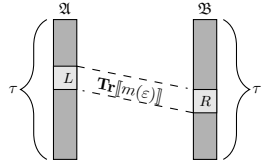
\begin{figure}[h]
		\centering
		\begin{tikzpicture}[scale=0.62,transform shape]
			
				\node (R1A) {};
				\path (R1A)+(0.5,-3) node (R1B) {};
				\draw [fill=black!30] (R1A) rectangle (R1B);
				
				\path (R1A)+(0,-1) node (R1sA) {};
				\path (R1sA)+(0.5,-0.5) node (R1sB) {};
				\draw [fill=black!10] (R1sA) rectangle (R1sB);
				
				\path (R1A)+(3,0) node (R2A) {};
				\path (R1B)+(3,0) node (R2B) {};
				\draw [fill=black!30] (R2A) rectangle (R2B);
				
				\path (R1sA)+(3,-0.5) node (R2eA) {};
				\path (R1sB)+(3,-0.5) node (R2eB) {};
				\draw [fill=black!10] (R2eA) rectangle (R2eB);
				
				\path (R1sA)+(0.5,0) edge[-,dashed] node[sloped,anchor=center,below]{$\semaNoh{\module(\e)}$} (R2eA);
				\path (R2eA)+(0,-0.5) edge[-,dashed] (R1sB);
				
				\path (R1sA)+(0.08,-0.25) node[right] {\small{$L$}};
				\path (R2eB)+(0,0.25) node[left] {\small{$R$}};
				
				\draw [decorate,decoration={brace,amplitude=10pt},xshift=-4pt,yshift=0pt] (0,-3) -- (0,0) node [black,midway,left,xshift=-0.5cm] {$\tau$};
				\draw [decorate,decoration={brace,amplitude=10pt},xshift=4pt,yshift=0pt] (3.5,0) -- (3.5,-3) node [black,midway,right,xshift=0.5cm] {$\tau$};
				
				\path (0.25,0) node [above] {$\strA$};
				\path (3.25,0) node [above] {$\strB$};
		\end{tikzpicture}
		\caption{\small  One of possible transitions  $(\strA,\strB) \in \semaNoh{\module(\e)}$  of an atomic module $m(\e)$. 
The module ``checks'' if the relation $L$ (or a set of such relations) is true in the $\tau$-structure $\strA$ on the left, and, if yes, it updates the value of the ``register'' $R$ in  the $\tau$-structure $\strB$ on the right.	The interpretation of all other relations 	in $\tau$ remains unchanged, by inertia.}\label{fig:Inertia}
	\end{figure}

\vspace{-2ex} 

\subsection{Atomic Modules as Vector Transformations}

We  restrict atomic transductions so that the only relations that change from structure to structure are \emph{unary} and \emph{contain one domain element at a time}, similarly to  registers in a Register Machine \cite{ShepherdsonSturgis:1961}. Please see Section \ref{sec:machine-model} for an explicit discussion of the machine model used.

Register values,  i.e., the content of unary singleton-set relations, get updated   by state-to-state transitions. Updating the registers can be thought of re-interpreting a fixed set (a vector) of $k$ constants,  interpreted by domain elements. Such an update can be specified in any way. For example, it can be a linear algebra transformation $\bb=L\ba$, or any other transformation $\bb=T(\ba,{\rm EDB})$, that takes into account the input EDB relations, in general.

 A vector $\ba=R_1 \cdots R_k$ of registers (one-dimensional array) can encode an $l$-dimensional matrix by providing appropriate  indices such as  $i,j,l$, e.g., $R_{i,j,l}$ for a 3-dimensional matrix. These indices correspond to keys $i,j,l$ in $([i,j,l]; \ R)$, a Graph Normal Form (GNF) representation of the 3-dimensional  matrix. 

Note that EDB relations given on the input are never updated. All computational work happens  in the registers.

\subsection{Intuitions for a Specific Secondary Logic}

As mentioned above, the bottom level of the algebra can be specified in any formalism. 
In the sections that follow, we specify a particular secondary logic for axiomatizing atomic transductions, that is similar to Conjunctive Queries (CQs).

 Intuitively, we take a \emph{monadic} primitive positive (M-PP) relation, i.e., 
 a unary relation definable by a unary CQ,
and output, arbitrarily, only 
 a \emph{single} element contained in the relation, 
instead of the whole relation. There could be several such CQs in the same modules, applied ot once. A formal definition will be given shortly.
 
 \vspace{1ex}
 
For readers familiar with Datalog (such a reader is invited to jump ahead to see the general form (\ref{eq:SM-PP-module})), we mention informally that, at the  bottom level  of the algebra, we have a set of Datalog-like ``programs'', one per each atomic module symbol $\module\in \Modules$. The ``programs''  are similar to Conjunctive Queries (non-recursive Datalog programs). The rules in such programs have a unary predicate symbol in the head of each rule. 
	Each (simultaneous) application of the rules  puts only \emph{one} domain  element into each unary IDB relation in the head, out of several possible.  This creates non-determinism in the atomic module applications.
The applications of the modules are controlled by algebraic expressions in (\ref{eq:algebra}) at the top level of the algebra.

This logic  is the second parameter that affects the expressive power of the language, in addition to the selection of the algebraic operations  presented earlier. The logic is carefully constructed to ensure complexity bounds presented towards the end of the paper.

\subsection{Preliminaries: Conjunctive Queries}\label{subsec:CQs}

We start our formal exposition  by reviewing queries, Conjunctive Queries(CQs), and PP-definable relations. 

Let $C$ be a class of relational structures of some vocabulary $\tau$.  Following Gurevich \cite{Gurevich-challenge}, we say that  an $r$-ary \emph{$C$-global relation $q$} (a \emph{query}) assigns to each structure $\strA$ in $C$ an $r$-ary
	relation $q(\strA)$ on $\strA$; the relation $q(\strA)$ is the specialization of $q$ to $\strA$. The vocabulary $\tau$ is
	the vocabulary  of $q$. If $C$ is the class of all  $\tau$-structures, we say that $q$ is
	\emph{$\tau$-global}. 
Let $\mathcal{L}$ be a logic. A $k$-ary query $q$ on $C$ is $\mathcal{L}$-definable if there is
an $\mathcal{L}$-formula $\psi(x_1, \dots , x_k)$ with $x_1, \dots , x_k$ as
free variables  such that for every $\strA \in  C$, 
\begin{equation}\label{eq:query-output}
q(\strA) = \{(a_1, \dots , a_k) \in \Domain^k \mid  \strA \models \psi(a_1, \dots , a_k)\}.
\end{equation}
In this case, $q(\strA)$ is called the \emph{output of $q$ on $\strA$}.
Query $q$ is \emph{unary} if $k=1$.
A \emph{conjunctive query} (CQ) is a query definable by a FO formula in prenex normal form built
from atomic formulas, $\land$, and $\exists$ only. A relation is   \emph{Primitive Positive (PP)} if it is definable by a C\textbf{\textbf{\textbf{A:}}}
$$
\forall x_1 \dots  \forall x_k \ \big( R(x_1, \dots  , x_k)  \leftrightarrow \underbrace{ \exists z_1 \dots \exists z_m\ (B_1 \land \cdots \land B_m)}_{\Phi(x_1, \dots  , x_k)} \big),
$$
and each atomic formula $B_i$ has object variables from 
$x_1, \dots  , x_k, z_1, \dots , z_m$.  We will use $\bx$, $\vu$, etc., to denote tuples of variables, and use the broadly used Datalog notation for the sentence above:
\begin{equation}\label{eq:CQ-Datalog}
	R(\bx) \leftarrow \underbrace{B_1(\vu_1), \dots, B_n(\vu_n)}_{\Phi_{\rm body}},
\end{equation}
where the part on the right of $\leftarrow$ is called the \emph{body}, and the part to the left of $\leftarrow$ the \emph{head} of the rule or the \emph{answer}. 
Notice that  an answer  symbol is always different from the symbols in the body  (\ref{eq:CQ-Datalog}), since, by definition, CQs are not recursive. 
We say that $R(x_1, \dots  , x_k)$ is  \emph{monadic} if $k=1$, and apply this term to the corresponding PP-relations as well.
\begin{example}\label{ex:Path2}
{\rm	Relation \textbf{Path of Length Two} is PP-definable, but not monadic: 
$P(x_1, x_2) \leftarrow E(x_1, z), E(z, x_2)$.
} 
\end{example}
\begin{example}\label{ex:Reach2}
	{\rm Relation	\textbf{At Distance Two  from $X$} is monadic and PP-definable:
$D(x_2) \leftarrow  X(x_1),E(x_1, z), E(z, x_2)$.
	} 
\end{example}
From now on, we assume that all queries are conjunctive, and their answers are monadic,  as in Example \ref{ex:Reach2}.

We now introduce a \emph{modification of CQs} we use as the secondary logic, the bottom level of our algebra.
Consider a relational vocabulary 
$
\tau \ := \ \tau_{\rm EDB}   \uplus 
 \tau_{\rm reg},
$
consisting of two disjoint sets of relational symbols. The symbols in $\tau_{\rm reg}$ are monadic, and $\rm reg$ abbreviates ``registers''.  Their interpretations change during a computation formalized by an algebraic term. The arities of the symbols of $\tau_{\rm EDB}$  are arbitrary. 
Their interpretations  are given by the input structure. 

\subsection{SM-PP Atomic Modules}

Let $	R_i'(x) \leftarrow  \Phi_{\rm body}$ be a unary CQ where $	R_i \in \tau_{\rm reg}$ and $\Phi_{\rm body}$ is in $\tau$. The reason for the renaming of $R_i$ by $R'_i$ is to specify $R_i$'s value in a successor state, without introducing recursion over $R_i$, which is not allowed in CQs.

\begin{definition}\label{def:SM_PP}\rm
	\emph{Singleton-set-Monadic  Primitive Positive ({\rm SM-PP}) relation} is a singleton-set relation $R'$  implicitly definable by:
\vspace{-1ex}	\begin{equation}\label{eq:SM-PP}
		\begin{array}{c}
			\forall x \ \big(  
		( \Phi_{\rm body}(x)\to 	R_i'(x) )\\
			\land \     \forall x \forall y (R_i'(x)\land R_i'(y) \to x=y) \big)\big).
		\end{array}
	\end{equation}
\end{definition}
\begin{notation}
	We use a rule-based notation for (\ref{eq:SM-PP}):
	\begin{equation}\label{eq:SM-PP-rule}
		R_i(x) \rul \Phi_{\rm body}.
	\end{equation}
\end{notation}
\noindent Notation ``$\rul$'' in (\ref{eq:SM-PP-rule}), unlike Datalog's ``$\leftarrow$'' in (\ref{eq:CQ-Datalog}), is used to emphasize that only one domain element is put into the relation in the head of the rule.

\begin{example}\label{ex:Reach2-rule}
{\rm	Suppose we want to put just one (arbitrary) element in extension of  $D$  denoting ``At Distance Two'' from Example \ref{ex:Reach2}. In the rule-based 
	syntax (\ref{eq:SM-PP-rule}):
	$$	D(x_2) \rul X(x_1), E(x_1, z) , E(z, x_2).$$
	
	\noindent	The defined relation is SM-PP. Since only one domain element, out of those at distance two from $X$, is put into $D$, there could be multiple outcomes, up to the size of the input domain.	
}
\end{example}

\begin{example}
{\rm	The same example ``At Distance Two'' can be formalized in a variant of Codd's relational algebra, as a  join of several relations, with a new operation $\pi^{\e}_{\rm attribute} $   of  \emph{Choice-projection}, where one element is put into the unary answer $D$. Here, the attribute is $x_2$.
	$$
D \ :=\ \pi^{\e}_{x_2}  \ 	(X \Join_{x_1} E \Join_{z} E).
	$$
}
\end{example}

\begin{example}{\rm A non-example of SM-PP relation is Path of Length Two  (Example \ref{ex:Path2}). The relation $P$  on the output of the atomic module is neither monadic, nor singleton-set.
	} 
\end{example}

\begin{definition}\label{def:SM-PP-module}
	An \emph{SM-PP module}  is a set of rules of the form (\ref{eq:SM-PP-rule}):
	\begin{equation}\label{eq:SM-PP-module}
	\module(\e) \ := \ \defin{ 
		R_1(x) \rul \Phi^1_{\rm body} \\
		\cdots \\
		R_k(x) \rul \Phi^k_{\rm body}
	},
\end{equation}
where $\module(\e)$ is a symbolic notation for the module, $\module \in \Modules$, and $\e$ is a free function variable ranging over Choice functions.
\end{definition}

\noindent In the rest of this paper, the modules are SM-PP, that is, they define SM-PP relations. We use modules to specify atomic transitions or \emph{actions}. 
Each module may update the interpretations of several registers simultaneously.  Thus, we allow limited parallel computations. The parallelism is essential --
this way we avoid using relations of arity greater than one, such as in the Same Generation example, explained in Section \ref{sec:Examples}.

\subsubsection{Inertia for SM-PP Modules}

We now formulate the Law of Inertia  for  SM-PP modules. Consider a module specification in the form (\ref{eq:SM-PP-module}).
We say that a value $a$ in a successor state $\strB$ of a register $R_i$ (where $a$ is a domain element) is \emph{forced by transition} $
(\strA,\strB) \in \semaNoh{\module(\e)}$  if, 
$$
\text{whenever } \strA \models \Phi^i[a/x], \text{we have that } a \in R_i^{\strB}.
$$
Here, $\Phi^i[a/x]$ is a formula 
in the secondary logic (here, a unary conjunctive query with $x$ being a free variable). 
These queries must be applied  in $\strA$ simultaneously,  in order for the update $
(\strA,\strB) \in \semaNoh{\module(\e)}$  to happen.

\begin{definition}\label{def:module-semantics}
The \emph{transitions specified by $\module(\e)$} is a binary relation $\semaNoh{\module(\e)} \subseteq \Un \times \Un$ such that, for all transitions  $(\strA,\strB) \in \semaNoh{\module(\e)}$, the values of all registers in $\strB$, forced by the transition, are as specified by the rules (\ref{eq:SM-PP-module}), and  
all other relations remain unchanged in the transition from $\strA$ to $\strB$.
\end{definition}

	\begin{definition}\label{def:Lo}
		By $\Lo$ we denote the algebra with a two-level syntax  that consists of the algebra defined in Section \ref{sec:Algebra},
		with SM-PP modules of the form (\ref{eq:SM-PP-module}), see Definition \ref{def:SM-PP-module}.
	\end{definition}

\subsection{Machine Model} \label{sec:machine-model} 

We can think of evaluations of algebraic expressions  as  computations  of Register machines starting from input $\strA$. 
The machines are reminiscent those of Shepherdson and Sturgis \cite{ShepherdsonSturgis:1961}.
Importantly, we are interested in \emph{isomorphism-invariant} computations of these machines, i.e., those that do not distinguish between isomorphic input  structures. 
Intuitively, we have: 

 \begin{compactitem}
	\item monadic ``registers'' -- predicates used during the computation, each containing only one domain element at a time; 
	
	\item  the ``real'' inputs, e.g., the edge relation $E(x,y)$ of an input graph, are of any arity;
	\item atomic transitions correspond to conditional assignments with a non-deterministic outcome;
	\item in each atomic step, only the registers of the previous state or the input structure are accessible;	
	\item  a concrete Choice function, depending on the history,
	chooses one of the possible outputs;

	\item computations are  controlled by  algebraic terms that, intuitively, represent programs.
\end{compactitem}

\vspace{1ex} 

In Section \ref{sec:Dynamic-Logic}, we saw that the main programming constructs are definable. Thus, programs for our machines are very close to standard Imperative programs. In Section \ref{sec:Examples}, we give examples of ``programming'' in this model.

\section{Examples}\label{sec:Examples}


We now give some cardinality, reachability examples, and examples with mixed propagations. 
We assume that the input structure $\strA$ is of combined vocabulary $
	\tau \ := \ \tau_{\rm EDB}   \uplus 
	\tau_{\rm reg}
	$, where $\strA|_{\tau_{\rm EDB}}$ is an input to a computational problem, e.g., a graph,
		and the ``registers'' in 	$\strA|_{\tau_{\rm reg}}$ are interpreted by a ``blank'' symbol. We use $\alpha$ with subscripts to denote algebraic terms, and add $\tau_{\rm EDB}$ symbols, e.g., $P$ and $Q$ in $\alpha_{\rm eq\_size}({P},{Q})$, to emphasize what is given on the input.

\subsection{Cardinality Examples}

\subsubsection{Size Four}

 \vspace{2mm}
 
 \noindent \fbox{\parbox{\dimexpr\linewidth-2\fboxsep-2\fboxrule\relax}{ Problem: \textbf{Size Four} $\alpha_4$
 		
 		Given: A structure $\strA$ with a  vocabulary symbol $adom$ denoting its active domain.\\
 		Question: Is $|adom^\strA|$ equal to 4?}}
 \smallskip 
 $$
 \GuessP(\e)  : =    \defin{  P(x) \rul adom(x) }.
 $$
 Here, we put an arbitrary element of the active domain into $P$. We specify guessing a new element by checking that the interpretation of $P$ now has never appeared in the trace of the program before:
 $$
 \begin{array}{l}
 \GuessNewP(\e)   : = 
\GuessP  \comp   \BG (P_{now} \not \eq P ). 
 \end{array}
 $$
 The problem Size Four is now specified as:
 $$
 \alpha_4:= {\GuessNewP\,}^4   \comp \rneg \GuessNewP,  
 $$
where the power four means that we execute the guessing procedure four times sequentially.
 The answer to the question $
 \strA \models | \alpha_4 \rangle \Last(\e),
 $
 is non-empty, i.e., it is possible to find a concrete Choice function to semantically instantiate $\e$, if and only if the input domain is of size four. 
 Obviously, such a program can be written for any natural number.

 \subsubsection{Same Size}
\vspace{2mm}
 
 \noindent \fbox{\parbox{\dimexpr\linewidth-2\fboxsep-2\fboxrule\relax}{ Problem: \textbf{Same Size} $\alpha_{\rm eq\_size}({P},{Q})$
 		
 		Given: Two unary relations $P$ and $Q$. \\
 		Question: Are $P$ and $Q$ of the same size?}}%
 \medskip

\noindent  We pick, simultaneously,  a pair of elements  from the two input sets, respectively: 
 $$
 \PickPQ(\e)  : = \defin{ \Pick_P(x ) \rul {P(x)}
 ,\ \ Pick_Q(x) \rul {Q(x)}
 }.
 $$
Store the selected elements temporarily:
$$
\Copy (\e)  : =    \defin{ P'(x ) \rul {\Pick_P(x )}, \ \
 Q'(x ) \rul {\Pick_Q(x )}
}.
$$
Here, the Choice variable $\e$ is not really necessary, since the module is deterministic. Next, we define a sub-procedure.
$$
\GuessNewPair  : = 
\big(\PickPQ  \comp   \BG (Pick_P \not \eq  P')  \comp \Pick_Q \not \eq  Q'  )\big) \comp   \Copy.
$$
In the sub-procedure above, we guess two new elements, one from each set, simultaneously.
The problem Same Size is now specified as:
$$
\alpha_{\rm eq\_size} ({P},{Q}):=  (\GuessNewPair) ^\iter \comp \rneg \PickPQ.
$$
The answer to the question $
\strA \models | \alpha_{\rm eq\_size} \rangle \Last(\e)
$
is non-empty, i.e., there is a Choice function witnessing $\e$, if and only if the extensions of predicate symbols in the input structure $\strA$ are of equal size.

 \subsubsection{EVEN}
 \vspace{2ex}

 \noindent \fbox{\parbox{\dimexpr\linewidth-2\fboxsep-2\fboxrule\relax}{ Problem: \textbf{EVEN} $\alpha_E $
 		
 		Given: A structure $\strA$ with a  vocabulary symbol $adom$ denoting its active domain.\\
 		Question: Is $|adom^\strA|$ even?}}
 \medskip 
 
 EVEN is PTIME computable, but is not expressible in Datalog, or any fixed point logic, unless a linear order on the domain elements is given. It is also known that Monadic Second-Order (MSO) logic over the empty vocabulary cannot express EVEN.

  We now show how to axiomatize it in our logic.
 We construct a 2-coloured path in the transition system by guessing new domain elements one-by-one, and using $E$ and $O$ as labels.\footnote{It is possible to use fewer register symbols, but we are not trying to be concise here.}
 To avoid infinite loops, we make  sure that the elements  never repeat.
 We define three atomic modules. For the deterministic ones, we omit the Epsilon variable:
 $$
 \begin{array}{lcl}
 	GuessP(\e) & : = &  \defin{  P(x) \rul adom(x) },
 	\\
 	CopyPO & : = & \defin{O(x ) \rul P(x)},\\
 	CopyPE & : = & \defin{ E(x ) \rul  P(x)}.
 \end{array}
 $$
 $$
 \begin{array}{l}
 	GuessNewO  : = \\
 \ \ \ \ \ \ \ \ \ \ \ \ \ 	\big(GuessP  \comp   \BG (P \not \eq  E ) \BG   (P \not \eq O ) \big)\ \comp \  CopyPO,\\
 	GuessNewE  : = \\
  \ \ \ \ \ \ \ \ \ \ \ \ \ 		\big(GuessP  \comp   \BG (P \not \eq E ) \BG   (P \not \eq O)\big)\ \comp \ CopyPE.
 \end{array}
 $$
 The problem EVEN is now formalized as:
$$ 
 \alpha_E:= (GuessNewO \comp  GuessNewE ) ^\iter \comp \rneg GuessNewO.
 $$
 The program is successfully executed if each chosen element is  different from any elements selected so far in the current information flow, and if $E$ and $O$ are guessed in alternation.
 Given a structure $\strA$ over an empty  vocabulary, the result of the query
 $
 \strA \models | \alpha_E \rangle \mT(\e)
 $
is non-empty whenever there is a successful execution of  $\alpha_E$, that is, the size of the input domain is even.

 \vspace{3mm}

 \subsection{Reachability Examples}

\subsection{s-t Connectivity}

 \vspace{2ex}
 
\noindent
\fbox{\parbox{\dimexpr\linewidth-2\fboxsep-2\fboxrule\relax}{ 
		Problem: \textbf{s-t-Connectivity} $\alpha({E},{S},{T})$\\
		Given: Binary relation $E$,  two constants $s$ and $t$, as singleton-set relations $S$ and $T$. \\
		Question: Is $t$ reachable from $s$ by following the edges?}}
\medskip  

\noindent To specify the term encoding this problem, we use the definable  constructs of imperative programming defined in Section \ref{sec:programming-constructs}:
$$
\begin{array}{l}
\alpha({E},{S},{T}) :=  \ \  M_{base\_case} \comp
{\bf repeat \ }   \big(   M_{ind\_case}  \comp \\  \BG ({\Reach' \not \eq \Reach }  ) \big)  \comp  \Copy  
   \ \ {\bf until } \ {\Reach \eq T} .
\end{array}
$$
Here, we  use a unary relational symbol (a register) $\Reach$. Initially, the corresponding relation contains the same node as $S$. The execution is terminated when $\Reach$ equals $T$. Register $\Reach'$ is used as a temporary storage.
To avoid guessing the same element multiple times, we use the $ \BG$ construct. 
The atomic modules used in this program are:
$$
\begin{array}{c}
\begin{array}{lcl}
M_{base\_case}(\e) & : = &  \  \defin{   \Reach(x ) \rul {S(x)}},
\end{array} \\
\begin{array}{l}
M_{ind\_case} (\e)  : =   \defin{   \Reach'(y ) \rul {\Reach(x)}, {E(x,y)}},
\end{array}\\
\begin{array}{l}
\Copy (\e) : = 
{\ \ \  }  \ \defin{  \ \Reach(x ) \rul {\Reach'(x)}}.
\end{array}
\end{array}
$$
Here, module $M_{ind\_case}$  is the only non-deterministic module. The other two modules are deterministic (i.e., the corresponding binary relation is a partial function).
Given structure $\strA$ over a vocabulary that matches the input (EDB) predicate symbols $E,S$ and $T$, including matching the arities,   by checking 
$
\strA \models |\alpha \rangle \Last(\e),
$
we verify that there is a successful execution of  $\alpha$. That is, $t$ is reachable from $s$ by following the edges of the input graph.

 \subsection{Same Generation} 
 
\vspace{2ex}

\noindent \fbox{\parbox{\dimexpr\linewidth-2\fboxsep-2\fboxrule\relax}{

		{Problem: \textbf{Same Generation}  $\alpha_{\rm SG}({E},{\Root}, {A},{B})$}
		
		Given: Tree -- edge  relation: $E$;  root:   $\Root$;  two nodes represented by unary singleton-set relations:   $A$ and $B$ 
		
		Question: Do $A$ and $B$ belong to the same generation in the tree? }} 

\medskip

\noindent Note that,  since we do not allow binary  ``register'' relations (binary EDB relations are allowed),  we need to capture the notion of being in the same generation through coexistence in the same structure. 
		$$
\hspace{-2ex}	M_{base\_case}(\e)   : =  \defin{ \Reach_A(x ) \rul {A(x)}, \   
			\Reach_B(x ) \rul {B(x)}
	}.
		$$

			Simultaneous propagation starting from the two nodes:
		$$
		\begin{array}{l}
		\hspace{-2mm} M_{ind\_case}(\e)   : =  
	\defin{  \Reach_A'(x ) \rul {\Reach_A(y)}, {E(x,y)},\\
			\Reach_B'(v ) \rul {\Reach_B(w)}, {E(v,w)} }.
		\end{array}
		$$
		 This atomic module
		specifies that, if elements $y$ and $w$, stored  in the interpretations of $\Reach_A$ and $\Reach_B$ respectively, coexisted in the previous state,  then $x$ and $v$ will coexist in the successor state. We copy the reached elements into ``buffer'' registers:
		$$
		\begin{array}{l}
		Copy(\e)   : =   
		\defin{ \Reach_A(x ) \rul {\Reach_A'(x )},\ \ \ \\
			\Reach_B(x ) \rul {\Reach_B'(x )}}.
		\end{array}
		$$
   The resulting interpretation of $\Reach_A$ and $\Reach_B$ coexist in one  structure, which is a state in a transition system.
The algebraic expression, using the definable imperative constructs, is:
$$
\hspace{-2mm}\begin{array}{l}
\alpha_{\rm SG}({ E},{ \Root}, { A},{ B})   :=  M_{base\_case};
   {\bf repeat } 
  M_{ind\_case} \comp \Copy ; \\
\ \ {\ \bf until } \ ({\Reach_A \eq {\Root}}\comp {\Reach_B \eq {\Root}}).
\end{array}
$$
While this expression looks like an imperative program, it really is a \emph{constraint} on all possible Choice functions, each  specifying a particular  sequence of choices. 
The answer to the question
$
\strA \models | \alpha_{\rm SG} \rangle \Last(\e)
$
is non-empty
if and only if $A$ and $B$ belong to the same generation in the tree.

\subsection{Linear Equations mod 2}

\vspace{2ex}

\noindent \fbox{\parbox{\dimexpr\linewidth-2\fboxsep-2\fboxrule\relax}{
		
		{Problem: \textbf{mod 2 Linear Equations   $\alpha_{\rm F}$}}

	Given: system $F$ of linear equations  mod 2 over vars $V$ given by two ternary relations $\Eq_0$ and $\Eq_1$ 
		
		Question: Is $F$ solvable? }}

\medskip

\noindent We assume that $V^{\strA}$ is a set, and $\Eq_0^{\strA}$ and $\Eq_1^{\strA}$ are  relations, 
both given by an input structure $\strA$ with $dom(\strA)=V^{\strA}$. Intuitively, $V^{\strA}$ is a set of variables, and $(v_1,v_2,v_3) \in \Eq_0^{\strA}$ iff $v_1\oplus v_2\oplus v_3 =0$, and $(v_1,v_2,v_3) \in \Eq_1^{\strA}$ iff $v_1\oplus v_2\oplus v_3 =1$. Such systems of equations are an example of constraint satisfaction problem that is not solvable by $k$-local consistency checks. This problem is known to be closely connected to the construction by Cai et al. \cite{CFI92}, and is not expressible in infinitary counting logic, as shown by Atserias, Bulatov and Dawar  \cite{AtseriasBD09}.  Yet, the problem is solvable in polynomial time by Gaussian elimination. We use the dynamic $\e$ operator to arbitrarily pick \emph{both} an equation (a tuple in one of the relations) and a variable (a domain element), on which Gaussian elimination $\Elim$ is performed.
$$
\begin{array}{l}
\alpha_{\rm F}({{ Eq_0}},{{ Eq_1}},{{ V}})   :=  M_{base\_case} \comp
 {\bf repeat } 
   \Pick\_\Eq\_V   \comp \Elim \\  

 \hspace{3cm}{ \bf until } \ \rneg \Pick\_\Eq\_V.
\end{array}
$$
Then, we have 
$\strA \models | \alpha_{\rm F} \rangle \Last(\e) 
$ returns a non-empty set  iff  $F$ is solvable.

\vspace{1ex}

\subsection{Observations}

\noindent\textbf{Two Types of Propagations}
In Reachability examples (s-t-Connectivity, Same Generation), propagations follow \emph{tuples} of domain elements given by the input structure, from one element to another. 
In Counting examples (Size 4, or any fixed size, Same Size, etc.),
propagations are made arbitrarily, they are \emph{unconstrained}. 
In Mixed examples (mod 2 equations, CFI graphs),
propagations are  \emph{of both kinds}, and they interleave.
Such examples with mixed propagation are not possible to formalize in Datalog, or any fixed point logic.  We belive that an \emph{interleaving} of constraned and unconstrained propagations (as in our logic)  is needed to formalize CFI-like examples. 
	Constrained propagations	are of a ``reachability'' kind, i.e., propagations over tuples.   Unconstrained propagations	 are  of a basic ``counting'' type, i.e., propagations from ``before'' to ``after'' via an unconstrained choice from the active domain. 
	We believe that the lack of this feature is the reason of why adding \emph{just} counting to FO(FP) is not enough to represent all properties in \PTIME \cite{CFI92} 
	--- e.g., the algorithm for mod 2 Linear Equations needs to interleave constrained and unconstrained propagations. Adding counting by itself to fixed points cannot accomplish it.

\vspace{2ex}

\noindent\textbf{Choice-Invariant Encodings}
In the given implementation of s-t-Connectivity, a wrong guess of a path is possible. However, one can write  a depth-first search algorithm, where the order of edge traversal does not matter. Because of this invariance, the depth-first encoding can be evaluated in \PTIME. This is because the length of the computation is limited to be polynomial in the size of the input structure (more about it in Section \ref{sec:length-choice-function}).
Choice-invariance would not hold for Hamiltoian Path, where, because it is an NP-complete problem, the possibility of a wrong guess always exists.

\section{Structural Operational Semantics}
\label{sec:SOS}

Our goal is to develop an algorithm that, given a Choice function $\ch$,  finds an answer to the  main task $ \strA \models |\alpha\rangle \mT(\ch/\e)$. We represent the algorithm as a set of rules in the style  of Plotkin's Structural Operational Semantics \cite{Plot81}. The transitions describe ``single steps'' of the computation, as in the computational semantics \cite{Henn90}.

\vspace{1ex}
	
		\noindent	\textbf{Identity (Diagonal)}	$\id$:
	$$
	\frac{true}{( \id, \sstring) \longrightarrow (\id, \sstring)}  .
	$$
				\noindent	\textbf{Atomic Modules}  $\module(\e)$: 
		$$
		\frac{true}{( \module(\e), \sstring \cdot\strA) \longrightarrow (\id, \sstring \cdot\strA\cdot\strB)}  \mbox{ if }    (\sstring \cdot\strA  \mapsto \sstring \cdot\strA\cdot\strB) \in \ch(\module(\e)).
		$$

\vspace{1ex}
		
		\noindent	\textbf{Sequential Composition} $\alpha \comp \beta$:
			$$
		\frac{(\alpha, \sstring ) \longrightarrow (\alpha',\sstring') }{(\alpha \comp \beta, \sstring) \longrightarrow
			(\alpha' \comp \beta, \sstring' ) }, \ \ \ \ \ \ \ \ \ 
		\frac{(\beta, \sstring ) \longrightarrow (\beta',\sstring')  }{(\id \comp \beta, \sstring ) \longrightarrow
			(\id \comp \beta', \sstring' ) }.	
		$$

		\noindent	\textbf{Preferential Union}  
			$\alpha \sqcup \beta$:
	$$
	\frac{(\alpha, \sstring ) \longrightarrow (\alpha',\sstring')   }{(\alpha \sqcup  \beta, \sstring ) \longrightarrow
		(\alpha' , \sstring') }.
	$$
	That is, $\alpha \sqcup  \beta$  evolves according to the instructions of  $\alpha$, if  $\alpha$ can successfully evolve to $\alpha'$. 
	$$
	\frac{(\beta, \sstring ) \longrightarrow (\beta',\sstring') \mbox { and  } (\rneg \alpha, \sstring ) \longrightarrow (\rneg \alpha,\sstring )    }{(\alpha \sqcup  \beta, \sstring ) \longrightarrow
		( \beta', \sstring') }  .
	$$
The rule says that $\alpha \sqcup  \beta$  evolves according to the instructions of  $\beta$, if $\beta$ can successfully evolve, while $\alpha$ cannot.		
		
	\vspace{1ex}	
		
			\noindent	\textbf{Right Negation (Anti-Domain)} $\rneg \alpha$:
	There are no one-step derivation rules for $\rneg \alpha$. Instead, we try, step-by-step, to derive $\alpha$, and, if not derivable, make the step.	
					$$
			\frac{true    }{(\rneg \alpha, \sstring) \longrightarrow ( \id,\sstring )  } \mbox{ if there is no Choice function $\ch'$ such that  derivation of $\alpha$ in $\sstring $ succeeds} .
			$$		
				\noindent	\textbf{Equality Check}	 	$(P\eq Q)$:
$$
\frac{true  }{( (P\eq Q), \sstring \cdot\strB) \longrightarrow
	(\id , \sstring \cdot\strB) } \mbox{ if }  P^{\strB} =  Q^{\strB}.
$$
 	
	\noindent	\textbf{Back Globally} $\BG(P\neq Q)$:

		$$
\frac{true  }{( \BG(P\neq Q), \sstring \cdot\strB) \longrightarrow
	(\id , \sstring \cdot\strB) } \mbox{ if for all } i, \  P^{\sstring(i)} \neq  Q^{\strB}.
$$
Here, $\sstring(i)$ is the $i$'th letter of $\sstring$, $\ffirst(\sstring)\leq i \leq \llast(\sstring)$.

	\noindent	\textbf{Maximum Iterate} $\alpha^\iter$:
		$$
	\frac{(\alpha, \sstring) \longrightarrow (\alpha',\sstring' )}{(\alpha^\iter, \sstring) \longrightarrow
		(\alpha' \comp \alpha^\iter, \sstring') }, \ \ \ \ \ \ \  \frac{true}{(\alpha^\iter, \sstring) \longrightarrow
		(\id, \sstring) } \mbox{ if $\rneg\alpha$ succeeds in $ \sstring$}  .
	$$
	Thus, $\alpha^\iter$ evolves according to $\alpha$, if $\alpha$ can evolve successfully, or if  $\rneg\alpha$ succeeds in $\sstring$.

\vspace{1ex}	
	
The execution of the algorithm consists of ``unwinding'' the term $\alpha$, starting with an input structure $\strA$.	The derivation process is deterministic: whenever there are two rules for a connective, only one of them is applicable. 
The goal of evaluation is to apply the rules of the structural
operational semantics starting from $(\alpha, \strA)$ and then
tracing the evolution of $\alpha$  to the ``empty'' program $\id$ step-by-step, by completing the branch upwards that justifies the step.
In that case, we say that 	the derivation  \emph{succeeds}. Otherwise, we say that it \emph{fails}.

\begin{proposition}
The evaluation algorithm based on the structural operational semantics, that   finds an answer to the  main task $ \strA \models |\alpha\rangle \mT(\ch/\e)$, 	is correct with respect to the semantics of $\Lo$.
\end{proposition}
\begin{proof}(outline)
The correctness of the algorithm follows by  induction on the structure of the algebraic expression, since the rules simply implement the semantics given earlier. 
\end{proof}

\section{Complexity of Query Evaluation}
\label{sec:Complexity}

In this section, we discuss the complexity of query evaluation in logic $\Lo$, specifically focusing on its fragment that captures precisely the complexity class NP.

	\subsection{The Length of a Choice Function} 
	\label{sec:length-choice-function}

Recall that we are interested in data complexity, where the formula is fixed. For an iterate-free term (without Maximum Iterate) and a fixed Choice function $\ch$, the main task is clearly in \PTIME.  But, we want to understand the data complexity of query evaluation for more general terms, under some natural restrictions.
	
	Recall that, according to (\ref{eq:terms-partial-maps}), terms are mapped to \emph{partial} functions on strings. Such a function can be continuously defined for a consecutive sequence of structures in $\Un$, but undefined afterwards. 
	
	By the \emph{length} of a Choice function we mean the maximal length of an ``extension'' string, i.e., a delta, that corresponds to the  mapping:
	$$
	length(\ch) \ : = \ \max \{ j-i \mid  (\sstring_{i} \mapsto \sstring_{j}) \in \ch(\module(\e))   \text{ for some } \module(\e)\}.
	$$
	Intuitively, it corresponds to 
	the height of the tree in Figure  \ref{fig:tree}.

	\begin{definition}
		We say that a Choice function $\ch$ is \emph{polynomial} if 
		$
		length(\ch) \in O(n^k),
		$
		where $n$ is the size of the domain of $\strA \in \Un$, and $k$ is some constant.
	\end{definition}
	\noindent 
Since Choice functions are certificates, their length has to be restricted. For our complexity results, we will restrict Choice functions to be polynomial.

\subsection{Data Complexity of Query Evaluation}

 \emph{Data complexity} refers to the case where the formula is fixed, and input structures $\strA$  vary \cite{Vardi82}. 
Here, we assume that the description of the transition system, specified in the secondary logic SM-PP (\ref{eq:SM-PP}), is a part of the formula in the Main Task (\ref{eq:main-task}), and, therefore, is fixed.
 
 We use the algorithm based on the structural operational semantics from Section \ref{sec:SOS} to analyze the data complexity of the main query $ \strA \models |\alpha\rangle \mT(\e)$.  The complexity depends on the nesting of the implicit quantifiers on Choice functions (cf. \ref{eq:quantifiers}), i.e., 
 how exactly $\rneg$ is applied in the term,  including as part of the evaluation of Maximum Iterate.

 Since the implicit (existential and universal) quantifiers can alternate, a problem at any level of the Polynomial Time Hierarchy can be expressed. 
Thus, the upper bound is PSPACE in general.

We can restrict the language so that $\rneg$  is applied to atomic modules only, including in Maximum Iterate, and there is no other application of negation, except the double negation needed to define the modality (the domain of a term).
For this restricted language, we can guarantee computations in NP because, in that case,  the size of the certificate we guess is polynomial, and it can verified in deterministic polynomial time.

	\begin{theorem} \label{th:membership-PTIME}
		The data complexity of checking $ \strA \models |\alpha\rangle \mT(\e)$ for a restriction of logic $\Lo$,  where negation applies to atomic modules only, is in \emph{NP}.
	\end{theorem}
	\begin{proof}(outline) We guess a certificate, which is an equivalence class $[\ch]_{(\strA,\alpha)}$  of  Choice functions (a trace, see (\ref{eq:equiv-relation})), of a polynomial length in the size of $\strA$, to instantiate the free function variable $\e$.  	With such an instantiation, term $\alpha$ becomes deterministic.  SM-PP atomic modules (essentially, conjunctive queries) can be checked, with respect to this certificate, in \PTIME.  
		We argue, by induction on the structure of an algebraic term, using Structural Operational Semantics from Section \ref{sec:SOS}, that all operations, including negation, can be evaluated in polynomial time. We take into account that 
	 choice functions are restricted to be of polynomial length, and  the term is fixed. Moreover, the semantics of Maximum Iterate does not allow loops in the transition system $\sem{\cdot}$.
		Thus, we return ``yes'' in polynomial time if the witness  $[\ch]_{(\strA,\alpha)}$ proves that the answer to $ \strA \models |\alpha\rangle \mT(\ch/\e)$ is ``yes''; or ``no'' in polynomial time otherwise. 
	\end{proof}

\subsection{Simulating NP-time Turing Machines}

The complexity class NP has been in the centre of theoretical computer science research for a long time. 
Its \emph{logical} characterization  was given by Fagin. His celebrated  theorem \cite{Fagin:1974} states that the complexity class NP coincides, in a precise sense, with second-order existential logic.

We now prove a counterpart of the celebrated Fagin's theorem \cite{Fagin:1974} for a \emph{fragment} of  our logic $\Lo$.   Recall that logic $\Lo$ is introduced in Definition \ref{def:Lo}. 
This fragment is existential, in the sense discussed in Section \ref{sec:implicit-quantification}.

In this section, we demonstrate that our logic is strong enough to encode any polynomial-time Turing machine over unordered structures. 
Together with Theorem \ref{th:membership-PTIME}, these two properties show that the fragment precisely \emph{captures} NP.

\begin{theorem}
For every NP-recognizable class $\cK$ of structures, there is a sentence of logic $\Lo$ (where negation applies to atomic modules only)  whose models are exactly $\cK$.
\end{theorem}

\begin{proof}  We focus on the  query 	$\strA \models | \alpha_{\rm TM} \rangle  \mT(\e) $
from  Section \ref{sec:Dynamic-Logic} and outline such a construction. 
The main idea is that a linear order on the domain elements of $\strA$ is \emph{induced} by a path in a transition system that corresponds to a guessed Choice function $\ch$. In this path, we \emph{guess} new elements one by one, as in the examples. The linear order corresponds to an order on the tape of a Turing machine. After such an order is guessed, a deterministic computation, following the path, proceeds for that specific order. 
We assume, without loss of generality that the deterministic machine runs for $n^k$ steps, where $n$ is the size of the domain. The program is of the form:
$$
\alpha_{\rm TM} := 
 {\rm ORDER} \ \comp \ {\rm START}\  \comp \
  {\bf repeat \ } 
   {\rm STEP}  {\ \bf until } \ {\rm END}.
$$ 

Procedure ORDER: Guessing an order is perhaps the most important part of our construction. 
We use a secondary numeric domain with a linear ordering, and guess elements one-by-one, using a concrete Choice function $\ch$. 
We associate an element of 
the primary domain with an element of the secondary one, using co-existence in the same structure.  Each Choice function corresponds to a possible linear ordering.

Procedure START: 
This procedure creates  an encoding of the input $\tau_{\rm EDB}$-structure $\strA$ (say, a graph)
in a sequence of structures in the transition system, 
to mimic an encoding $enc(\strA)$ on a tape of a Turing machine. We use structures to represent cells of the tape of the Turing machine (one $\tau$-structure = one cell).  
The procedure follows a specific path, and thus a specific order  generated by the procedure START.
Subprocedure 
$
{Encode}(vocab(\strA), \dots, S_\sigma, \dots , \bar{P}, \dots)
$
operates as follows. In every state (= cell), it keeps the input structure $\strA$ itself, and adds the part of the encoding of $\strA$ that belongs to that cell.
The interpretations of $\bP$ over the secondary domain of labels provide cell positions on the simulated input tape. 
Each particular encoding is done for a specific induced order on domain elements, in the path that is being followed. 

In addition to producing an encoding, the procedure START sets the  state   of the Turing machine to be the initial state $Q_0$.  It also sets initial values for the variables used to encode the instructions of the Turing machine. 

Expression START is similar to the first-order formula  $\beta_{\sigma}(\bar{a})$ used by  Gr\"{a}del in his proof of capturing \PTIME using SO-HORN logic on ordered structures \cite{Graedel:1991}.  The main difference is that instead of tuples of domain elements $\bar{a}$ used to refer to the addresses of the cells on a tape, we use tuples $\bP$, also of length $k$. Gr\"{a}del's formula $\beta_{\sigma}(\bar{a})$  for encoding input structures has the  following property: $(\strA,<) \models \beta_{\sigma}(\bar{a}) \ \  \Leftrightarrow\ \  \mbox{ the $\bar{a}$-th symbol of $enc(\strA)$ is $\sigma$.}  $ 
Here, we have:   
$$
\begin{array}{c}
\strA \models Encode(\dots, S_\sigma, \dots , P_1(a_1), \dots, P_k(a_k), \dots )(\ch/\e) \\ \Leftrightarrow\ \  \mbox{ the $P_1(a_1), \dots, P_k(a_k)$-th symbol of $enc(\strA)$ is $\sigma$,}
\end{array} 
$$
where $\ba$ is a tuple of elements of the secondary domain, $\ch$ is a Choice function  that   guesses a linear order on the input domain through an order on  structures (states in the transition system),   starting in the input structure $\strA$. That specific generated  order is used in the encoding of the input structure. Another path produces a different order, and constructs an encoding of the input structure for that order.

Procedure STEP: This procedure encodes the instructions of the deterministic Turing machine. SM-PP modules are well-suited for this purpose. Instead of time and tape positions as arguments of binary predicates as in Fagin's \cite{Fagin:1974} and Gr\"{a}del's \cite{Graedel:1991} proofs, we use coexistence 
in the same structure with $k$-tuples of domain elements, as well as lexicographic successor and predecessor on such tuples.
Polynomial time of the computation is guaranteed because 
time, in the repeat-until part, is tracked with $k$-tuples of domain elements.

Procedure END: It checks if the accepting state of the Turing machine is reached.

We have that, for any \PTIME Turing machine, we can construct term $\alpha_{\rm TM}$ in the logic $\Lo$ such that the answer to 
 $
   	\strA  \models |\alpha_{\rm TM} \rangle  \mT(\e)
$
is non-empty 
if and only if the Turing machine accepts an encoding of $\strA$ for  some specific but arbitrary order of domain elements on its tape. 
\end{proof}
\noindent Combining the theorems, we obtain the following corollary.
\begin{corollary}
	The fragment of logic $\Lo$, where negation applies to atomic modules only,  captures precisely NP with respect to its data complexity.
\end{corollary}

\section{Related Work}
\label{sec:Related}

\noindent\textbf{Logic of Information Flows} 
The work in the current paper is a continuation of research initiated by the author, who introduced  the Logic of Information Flows (LIF) in  \cite{T:FROCOS:2019}. The goal of introducing LIF was to understand how information propagates, in order to make such propagations efficient.  
A version of LIF, \cite{T:FROCOS:2019},  was published, initially,  in the context of reasoning about modular systems, and was based on classical logic. In subsequent work,  with a group of coauthors,  we studied the notions of input, output and composition \cite{ABSTV:KR:2020}, \cite{LIF-TOCL} in LIF, and applications to data access in database research \cite{flifexfo}. A short description of that work can be found in \cite{T:KR-2020-short-abstract}. A lot of  development in the papers \cite{ABSTV:KR:2020}, \cite{flifexfo}, and beyond, was done in the excellent PhD work by Heba Aamer  \cite{Heba-thesis}. 
	
In parallel with the work on \cite{ABSTV:KR:2020}, \cite{LIF-TOCL}, the author continued working towards the main goal, on studying how  to make information flows efficient.  Defining LIF  based on classical logic, as was done in the early versions, was clearly not sufficient. One of the main observations of the author, very early in the development of LIF, was that it was necessary to  make operations \emph{function-preserving}.
The author conducted an extensive  analysis, going through numerous versions and combinations of algebraic operations, eventually coming up with a minimal, but sufficiently expressive set of the operations presented here.  
 Also, it was crucial to 
introduce  \emph{choice}, to handle atomic non-determinism.

\vspace{1ex}

\noindent\textbf{Choice Operator}   Choice occurs in many high-level descriptions of polynomial-time algorithms, e.g., in Gaussian elimination:  \emph{choose an element and continue.} A big challenge, in our goal of formalizing non-deterministic computations algebraically, in an algebra of partial functions,   was in how to deal with binary relations.  Such relations are not necessarily  functional. To deal with this challenge,  the author invented  history-dependent Choice functions, early in her work on multiple variants of LIF.  The dependence on the history, to the best of our knowledge,  has not been used in defining Choice functions before.
The first use of a  Choice operator $\e$ in logic goes back to Hilbert and Bernays in 1939  \cite{HilbertBernays:1939} for proof-theoretic purposes,  without  giving a model-theoretic semantics. 
Early approaches to Choice, in connection to Finite Model Theory,   
include the work by Arvind and Biswas \cite{ArvindBiswas87}, Gire and Hoang \cite{Gire-Hoang},  Blass and Gurevich \cite{BlassGurevich:Choice} and by  Otto \cite{DBLP:journals/jsyml/Otto00}, among others.   
Richerby and Dawar \cite{ Dawar-Richerby}  survey  and study logics with various forms of choice. 
Outside of Descriptive Complexity, 
Hilbert's $\e$ has been studied extensively by  Soviet logicians Mints, Smirnov and Dragalin in 1970's, 80's and 90's, see \cite{Soloviev17}. The semantics of this operator is still an active research area, see, e.g., \cite{Wirth17-long}. Unlike the earlier approaches, our Choice operator formalizes a \emph{strategy}, i.e., what to do next, given the history, under the given constraint given by a term.

An example of using strategies can be seen in building proofs
 in a Gentzen-style proof system. There, while selecting an element witnessing an existential quantifier, we ensure that the element is ``new'', i.e.,  has not appeared earlier in the proof.

A problem with a set-theoretic Choice operator  is that for first-order (FO) logic, and thus for its fixed-point extensions such as FO(FP), choice-invariance is undecidable \cite{BlassGurevich:Choice}. Therefore, in using FO, there is a danger of obtaining an undecidable syntax, which violates a basic principle behind a well-defined logic.
Choiceless Polynomial Time \cite{BlassGurevichShelahChoicelessPTIME} 
is an attempt to characterize what can be done \emph{without} introducing choice.
But, as a critical step, we
\emph{restrict FO connectives}, similarly to Description logics \cite{DescrLogic2003handbook-long}, 	and in strong connection to modal logics,  that are robustly decidable 
\cite{Vardi-Modal-Robust,Graedel-Modal-Robust}.

\vspace{1ex}

\noindent\textbf{Two-Variable Fragments}
The step towards binary relations in LIF \cite{T:FROCOS:2019}, where we partitioned the variables of atomic symbols into input and outputs, was inspired by our own work on Model Expansion \cite{MT05-long}, with its  before-after-a-computation perspective, and also by bounded-variable fragments of first-order logic.
Such fragments have been shown to have  good algorithmic properties by Vardi 	\cite{Vardi-Bounded-Variable-Queries}. Two-variable fragments have been offered as an initial explanation of the robust decidability of modal logics  
\cite{Vardi-Modal-Robust,Graedel-Modal-Robust}. 
In addition, two-variable fragments have order-invariance \cite{Zeume-long}.
Unfortunately, such fragments of FO are not expressive enough to  encode Turing machines.
To overcome this obstacle, we lifted the algebra to operate on \emph{binary relations on strings of relational structures}.  Such relations, intuitively, encode state transitions.

\vspace{1ex}

\noindent\textbf{Algebras of Binary Relations} 
Such an algebra was first introduced by De Morgan. 
It has been extensively developed by Peirce and then Schr\"{o}der. It was abstracted to relation algebra RA by J\'{o}nsson and Tarski in \cite{Jonsson-Tarski:1952}. For earlier uses of the operations and a historic perspective please see Pratt’s historic and  informative overview paper \cite{Pratt:calc-bin-rel}.
More recently, relation algebras were studied by Fletcher,  Van den Bussche,
Surinx and their collaborators in a series of paper, see, e.g. \cite{SVV17,FGLSVVVW}.
The algebras of relations consider various subsets of operations on binary relations as primitive, and other as derivable. 
Our algebra is an algebra of binary relations, similar to J\'{o}nsson and Tarski \cite{Jonsson-Tarski:1952}, however our relations are on  more complex entities  -- \emph{strings} of relational structures. 
Moreover, our binary relations are  \emph{functional}, which is achieved by making all operations function-preserving, and introducing Choice functions.

\vspace{1ex}

\noindent\textbf{Algebras of Functions} 
In another direction, Jackson and Stokes and then McLean \cite{JacksonStokes:2011,McLean17}   study partial functions and their algebraic equational axiomatizations.  The work of Jackson and Stokes \cite{JacksonStokes:2011} is particularly relevant, because it introduces some connectives we use.
However, they  do not study algebras on strings and Choice functions. Also, we had to eliminate intersection, which they, and many other researchers, use. We had to do it  because intersection wastes computational power due to confluence, and makes  information flows less efficient. We came up with an example where an intersection of (the representations of) two NP-complete problems produce a problem in  \PTIME.  We have selected a minimal set of operations, for our purposes. The operations correspond to dynamic  and function-preserving versions of conjunction, disjunction, negation and iteration.  However, other operations, such as many of those from  McLean \cite{McLean17}, can be studied as well.

\vspace{1ex}

\noindent\textbf{Restricting Connectives} 
Our algebraic operations are a
restriction of  first-order connectives, similarly to restrictions of such connectives in  Description logics \cite{DescrLogic2003handbook-long}. Classical connectives, such as negation, disjunction, and also the the iterator in the form of the Kleene star (reflexive transitive closure) are incompatible with an algebraic setting of \emph{functions}. We require the connectives to be function-preserving. We traced the origins of the operations we use as follows. 
The  Unary Negation  operation is from Hollenberg and Visser \cite{HollenbergVisser:99}.  It is also studied, among other operations on functions,  by McLean in, e.g.,  \cite{McLean17} in the form of the Antidomain operation.
Restricting negation (full complementation) to its unary counterpart is already known to imply good computational properties \cite{unary-negation-long}.\footnote{Unary negation is related to the negation of modal logics. It is the modal negation in the modal Dynamic Logic we discuss later. In general, modal logics are  known to be robustly decidable \cite{Vardi-Modal-Robust,Graedel-Modal-Robust} due to a combination of properties.} However, in our case, \emph{all} connectives  and the fixed point construct of first-order with least fixed point,  FO(LFP), had to be restricted.
The operations of Preferential Union and Maximum Iterate are from  Jackson and Stokes \cite{JacksonStokes:2011}. While these algebraic operations are more restrictive than those of Regular Expressions,   
they define the main constructs of Imperative programming.

\vspace{1ex}

\noindent\textbf{Database Query Languages and Comprehension Syntax} 
Early functional languages for databases include
Macchiavelli \cite{DBLP:conf/sigmod/OhoriBT89} and 
Kleisli  \cite{DBLP:journals/jfp/Wong00} 
The languages are based on the 	Comprehension Syntax proposed in \cite{DBLP:journals/sigmod/BunemanLSTW94}, see also \cite{DBLP:journals/tcs/BunemanNTW95} for a later development. 
In somewhat different line of research, Libkin proposed to use 
a polynomial space iterator, among other constructs, for querying databases with incomplete information \cite{DBLP:conf/dbpl/Libkin95}.

\vspace{1ex}

\noindent\textbf{Substitution Monoid and Connection to Kleene Algebra}
 First, we explain a connection to a Kleene algebra on a meta-level, and then discuss the differences in the languages.
The algebra forms a monoid with respect to term substitution.  The proof of this statement is lengthy and is outside of the scope of this paper. However, as for every monoid, there is a well-known and natural connection to a Kleene algebra.

\vspace{1ex}

Let $M$ be a monoid with identity element $\1$ and let $A$ be the set of all subsets of $M$. For two such subsets $S$ and $T$, let $S + T\ : = \  S \cup \ T$, and let 
$
ST \ : = \ \{st : s \in S\text { and } t \in T\}.
$
We define $ S^*$  as the submonoid of $M$ generated by $S$, which can be described as 
$
S^* \ : = \ \{\1\} \cup S  \cup  SS  \cup  SSS  \cup  ... 
$
Then $A$ forms a Kleene algebra with 0 being the empty set and 1 being $\{\1\}$.

\vspace{2ex} 

\noindent\textbf{Comparison to Kleene Algebra at the Level of Algebraic Languages}
The connection to a Kleene algebra we have just explained exists at the meta-level, when we consider term substitution as the monoid operation. But, what is the connection to Kleene algebra  at the level of the algebra itself, i.e., its algebraic operations? One difference is that, unlike the operations of Union ($\cup$) and Iteration ($^*$) of Regular Expressions,  the  operations of Preferential Union ($\sqcup$) and Maximum Iterate ($^\iter$) are \emph{function-preserving}. 

\vspace{1ex} 

\noindent\textbf{Temporal and Dynamic Logics} 
The algebra (\ref{eq:algebra})  has an alternative (and equivalent) syntax in the form of a Dynamic logic, that we explain shortly. Our Dynamic logic is fundamentally different from Propositional Dynamic Logic (PDL)  \cite{DBLP:conf/focs/Pratt76,DBLP:journals/jcss/FischerL79} in that, because of the Choice function semantics, it has a linear time (as opposite to branching time) semantics. This is similar to  to Linear Temporal Logic (LTL) and Linear Dynamic logic on finite traces  LDL$_f$   \cite{DeGiacomo-Vardi:IJCAI:013}.\footnote{LDL$_f$ has the same syntax as PDL, but is interpreted over traces.}
But, in addition,  branching from nodes is used for complex tests. So, there is some similarity to CTL$^*$. To the best of our knowledge, the data complexity of temporal and dynamic logics, with respect to an input database, has never been studied. 

\vspace{1ex} 

\noindent\textbf{Algebra vs Logic} 
Our algebra-Dynamic-logic connection is reminiscent that of  Kleene algebras with tests (KAT) \cite{Kozen:KAT:97}; but, unlike KAT that allows for simple tests only, our logic allows for arbitrarily complex nested tests, in general.  

We believe that ours is the first \emph{algebraic} formalization of a \emph{linear} time logic, i.e., a logic interpreted over traces of computation (as opposite to branching time logics with branching at every state).
A crucial step of this formalization is the use of  Choice functions that map strings to strings. We are not aware of any use of such functions in temporal or  Dynamic logics.

\vspace{1ex}

\noindent\textbf{Partial Algebras and Logics} 
Our approach to strong equality is inspired by Partial Horn Logic (PHL) by Palgrem and Vickers \cite{PalmgrenVickers07}.
The logic builds partiality directly into the logic. 
Their logic is, essentially, as in  \cite{Joh02b}, but has a modified substitution axiom. It identifies definedness with self-equality. The axioms of this logic are universal Horn formulae.
A quasi-equational theory in this partial logic has functions but no predicates (other than equality). Axioms are given in sequent form with conjunction of equations entailing an equation. 
We adopted this idea for reasoning about computation. But, in addition to definedness, we introduced undefinedness, that indicates non-existence of a ``yes'' certificate.
The use of partial algebras has a long history, and is surveyed 
in \cite{Burmeister-parial-algebras-intro}.

\section{Conclusion}
\label{sec:Conclusion}

We have defined  a query language in the form of an algebra of partial higher-order functions on strings of relational structures. 
The algebra  has a two-level syntax, where propagations are separated from control. 
The operations of the top level are obtained by taming (dynamic versions of) classical connectives and a fixed-point construct, that is, by making them \emph{function-preserving}. 
A particular example of a logic of the bottom level is  a singleton-set restriction of Monadic Conjuctive Queries that, intuitively, represent non-deterministic conditional assignments.

The algebra has an associated  declarative query language in the form of a Dynamic logic that is \emph{equivalent} to the algebra. In this logic, typical programming constructs, such as while loops and if-then-else, are definable.
In general, the logic can encode any Turing machines.  We have considered a restricted fragment where the length of the computation is limited to a polynomial number of steps, in  the size of the input structure.

Since the logic can implicitly mimic quantification over cetrificates, it can express problems at any level of the Polynomial Time Hierarchy. With  further restriction, where negations are applied to atomic modules only, the logic captures precisely the complexity class NP. 

We give examples expressing  counting properties on unordered structures,  even though the logic does not have a special cardinality construct. 
The logic can also express  reachability type of queries, as well as examples with mixed propagations.

A future step of this research is to understand \emph{under what general  conditions on the terms $\term$ of the logic  $\Lo$}, evaluating the main query 
$\strA \models\  |\term \rangle  \mT(\e)$ 
can be done by simply following one (arbitrary) sequence of atomic choices. 
When such an evaluation is possible, the query is  choice-invariant.

Work is under way on developing a proof system (and a quasi-equational theory)  for syntactic reasoning about strong equalities between terms,  including definedness, $\term(x)\dn$, and undefinedness, $\term(x)\!\notdn$, as  particular cases of strong equalities.

Another interesting direction is to see whether preservation theorems that fail in the case of function-preserving binary relations \cite{Bogaerts-ten-Cate-McLean-Van-den-Bussche-2023} would hold for the function-preserving binary relations \emph{over strings}, as in this paper. 

The proposed query language, to be practical, has to be extended to handle arithmetic and aggregates. It can be done using the framework of Embedded Model Expansion \cite{TM:IJCAI:2009,TT:NonMon30}, and also via defining a semiring semantics.

\section{Acknowledgements}

The author is grateful to Leonid Libkin and to Brett McLean for useful discussions on cardinality and reachability examples, and on algebras of partial functions, respectively. 
Many thanks to Heng Liu, Shahab Tasharrofi and Anurag Sanyal for their help with the figures. 
The author's  research  is supported by the Natural Sciences and Engineering Research Council of Canada (NSERC).
Part of the research presented in this paper was carried out while the author  
participated in the program on Propositional Satisfiability and Beyond of
the Simons Institute for the Theory of Computing during the spring of 2021, and its extended reunion. 
\bibliographystyle{plain}
\bibliography{bibliography}

\section{Appendix: Quasi-Equational Theory}

	In this section, we give a flavour of how the proof theory for our algebra will look like. In the proof system, we will be able to derive the existence and non-existence of a certificate, as explained in Section \ref{sec:Strong-Equalities}.
		First, we include all logical rules of Partial Horn Theory (PHT) \cite{PalmgrenVickers07}, which we are not going to reproduce here. 
			In addition to logical rules, we need a quasi-equational theory of the specific algebraic operations we use.
	
		Let a set $\Modules$ of module symbols be fixed. Let the signature $\sigma$ be the set of function symbols in the algebra (\ref{eq:algebra}), also listed in the table in Section \ref{sec:Algebra}.  
		
		A \emph{Horn theory} in a signature $\sigma$ is a set of Horn sequents over $\sigma$.  In the case where $\sigma$ does not contain any predicate symbols, the theory is called \emph{quasi-equational}.

A novelty of the  quasi-equational theory for our logic is that,  unlike PHT \cite{PalmgrenVickers07}, we compute not only defined terms $\term\dn$, but also undefined terms  $\term\notdn$.
		 Axioms for $t\notdn$  specify when each particular operation is not defined. Informally, $t(x)\notdn$ means that the corresponding program cannot be executed at $x$.

	\vspace{3ex} 
	
	Our specification of  the quasi-equational theory for signature $\sigma$ given in the table in Section \ref{sec:Algebra} is as follows. Variables that are universally quantified are listed over the turnstiles. Some turnstiles are bidirectional, meaning that implications hold in both directions.

\vspace{1ex}

\noindent\textbf{Identity:}
	\begin{align} 
		\top &\provesxx \id (x) = x 
	\end{align}

\vspace{1ex}

\noindent\textbf{Unary Negation:}
\begin{align} 
	y(x) \notdn   &\ \doubleprovesxy \ \rneg y(x) = x 
\end{align}
\begin{align}
		y(x) \notdn   &\ \doubleprovesxy \ \rneg y(x) \dn \\
	y(x) \dn   &\ \doubleprovesxy \ \rneg y(x) \notdn
\end{align}

\vspace{1ex}

	\noindent\textbf{Preferential Union:}
	\begin{align}  \label{eq:lneg}
	 y(x) =w  & \provesxyzw [y \sqcup z] (x) = w\\  
y(x) \notdn    \land \   z(x) =w & \provesxyzw  [y \sqcup z] (x) = w
	\end{align}
	\begin{align}
y(x) \notdn    \land \   z(x) \notdn   & \doubleprovesxyzw  [y \sqcup z] (x)\notdn \  \land \ [z \sqcup y] (x) \notdn 
	\end{align}

\vspace{1ex}

	\noindent\textbf{Maximum Iterate:}
	\begin{align} 
  y(x) \notdn  \ & \provesxy  y^{\iter}(x)=x\\
	 y^{\iter}(y(x)) \dn \ & \provesxy    y^{\iter}(x) = y^{\iter}(y(x))  
	\end{align}

For all algebraic operations, we need both, axioms for when the operation is defined, and when it is undefined.
If we try to write axioms to derive when Maximum Iterate is undefined, we immediately run into difficulties due to the cyclicity of the reasoning that is needed to prove that such a term is undefined. Stepan Kuznetsov has suggested a method to deal with that case.

\vspace{3ex}

In addition, we assume  axioms  compiled from  SM-PP atomic modules, for each atomic module symbol in $\Modules$.

\label{sec:quasi-equational-theory}

\end{document}